\documentclass[11pt,a4paper]{article}
 \usepackage{fullpage}
\usepackage{amsmath}
\usepackage{relsize}
\usepackage{amssymb}
\usepackage{wasysym}
\usepackage{stmaryrd}
\usepackage{graphicx}
\usepackage{psfrag}
\usepackage{color}
\usepackage{subfigure}

\newcommand{\id}{\mbox{\sf Id}}
\newcommand{\T}{\mbox{\sf T}}

\newcommand{\MF}{\mbox{\tt ME}}

\newcommand{\parent}{\mbox{\it p}}
\newcommand{\lab}{\mbox{\rm $\ell$}}
\newcommand{\size}{\mbox{\it size}}
\newcommand{\m}{\mbox{\it Out}}
\newcommand{\mIns}{\mbox{\it In}}
\newcommand{\N}{\mbox{\sf N}}

\newcommand{\Child}{\mbox{C}}
\newcommand{\SizeC}{\mbox{\rm SizeC}}
\newcommand{\NbrNdS}{\mbox{\rm nbrNdS}}

\newcommand{\LabelR}{\mbox{\rm Label$_R$}}
\newcommand{\LabelNd}{\mbox{\rm Label$_{Nd}$}}
\newcommand{\Label}{\mbox{\rm Label}}
\newcommand{\Isleaf}{\mbox{\rm Leaf}}
\newcommand{\Heavy}{\mbox{\rm Heavy}}
\newcommand{\Light}{\mbox{\rm Light}}

\newcommand{\Distance} {\mbox{\rm Distance}}

\newcommand{\OEd}{\mbox{\rm O\hspace*{-0,05cm}E}}
\newcommand{\IE}{\mbox{\rm I\hspace*{-0,01cm}E}}
\newcommand{\FC}{\mbox{\rm Cand$_c$}}
\newcommand{\FA}{\mbox{\rm Cand$_l$}}
\newcommand{\Fusion}{\mbox{\rm Candidate}}

\newcommand{\Lca}{\mbox{\rm $nca$}}

\newcommand{\path}{\mbox{\tt path}}

\newcommand{\RLC}{\mbox{\rm R$_{\tt Label}$}} 
\newcommand{\RSC}{\mbox{\rm R$_{\tt Size}$}} 
\newcommand{\RMin}{\mbox{\rm R$_{\tt Min}$}}
\newcommand{\RC}{\mbox{\rm R$_{\tt Rec}$}}
\newcommand{\RF}{\mbox{\rm R$_{\tt Merge}$}}
\newcommand{\RCo}{\mbox{\rm R$_{\tt Correct}$}}

\newcommand{\RDist}{\mbox{\rm R$_{\tt Dist}$}}
\newcommand{\RFEnd}{\mbox{\rm R$_{\tt End}$}}

\newcommand{\MSTA}{\mbox{\tt SS-MST}}
\newcommand{\LabA}{\mbox{\tt NCA-L}}

\newtheorem{lemma}{Lemma}
\newtheorem{corollary}{Corollary}
\newenvironment{proof}{\noindent\textbf{Proof.}}{ \bigskip\hfill$\Box$}
\newtheorem{definition}{Definition}
\newtheorem{theorem}{Theorem}

\newtheorem{specification}{Specification}

\begin{document}
\title{Fast Self-Stabilizing Minimum Spanning Tree Construction
Using Compact Nearest Common Ancestor Labeling Scheme \footnote{A preliminary version of this paper has appeared in the proceedings of the 24th International Conference on Distributed Computing (DISC 2010), see \cite{BlinDGR10}.}}

\author{
L{\'e}lia Blin\\
\footnotesize{Universit\'e d'Evry-Val d'Essonne, 91000 Evry, France}\\
\footnotesize{LIP6-CNRS UMR 7606, France.}\\
\footnotesize{lelia.blin@lip6.fr}
\and
Shlomi Dolev\\
\footnotesize{Department of Computer Science, Ben-Gurion University of the Negev, Beer-Sheva, 84105, Israel.}\\
\footnotesize{dolev@cs.bgu.ac.il}
\and
Maria Gradinariu Potop-Butucaru\\
\footnotesize{Universit\'e Pierre \& Marie Curie - Paris 6, 75005 Paris, France.}\\
\footnotesize{LIP6-CNRS UMR 7606, France.}\\
\footnotesize{maria.gradinariu@lip6.fr}
\and
St{\'e}phane Rovedakis\\
\footnotesize{Laboratoire CEDRIC, CNAM, 292 Rue St Martin, 75141 Paris, France.}\\
\footnotesize{stephane.rovedakis@cnam.fr}
}

\date{}
\maketitle

\begin{abstract}
We present a novel self-stabilizing algorithm for minimum spanning tree (MST) construction. 
The space complexity of our solution is $O(\log^2n)$ bits and it converges in $O(n^2)$ rounds. 
Thus, this algorithm improves the convergence time of previously known self-stabilizing asynchronous MST algorithms by 
a multiplicative factor $\Theta(n)$, to the price of increasing the best known space complexity by a factor $O(\log n)$. 
The main ingredient used in our algorithm is the design, for the first time in self-stabilizing settings, of a labeling scheme  
for computing the nearest common ancestor with only $O(\log^2n)$ bits. 
\end{abstract}

\thispagestyle{empty}
\setcounter{page}{0}
\newpage 

\section{Introduction}
\label{sec:intro}
Since its introduction in a centralized context~\cite{Prim57,Kruskal56}, the minimum spanning tree (or MST) problem gained a 
benchmark status in distributed computing thanks to the seminal work of Gallager, Humblet and Spira~\cite{GallagerHS83}.  

The emergence of large scale and dynamic systems revives the study of scalable algorithms.
A \emph{scalable} algorithm does not rely on any global parameter of the system (e.g. upper bound on the  number of nodes or the diameter). 

In the context of dynamic systems, after a topology change a minimum spanning tree previously computed is not necessarily a minimum one (e.g., an edge with a weight lower than the existing edges can be added). A mechanism must be used to replace some edges from the constructed tree by edges of lower weight. Park et al.~\cite{ParkMHT90,ParkMHT92} proposed a distributed algorithm to maintain a MST in a dynamic network using the Gallager, Humblet and Spira algorithm. In their approach, each node know its ancestors and the edges weight leading to the root in the tree. Moreover, the common ancestor between two nodes in the tree can be identified. For each non-tree edge $(u,v)$, the tree is detected as not optimal by $u$ and $v$ if there exist a tree edge with a higher weight than $w(u,v)$ between $u$ (resp. $v$) and the common ancestor of $u$ and $v$. In this case, the edge of maximum weight on this path is deleted. This yields to the creation of several sub-trees, from which a new MST can be constructed following the merging procedure given by Gallager et al.~\cite{GallagerHS83}. Flocchini et al. ~\cite{FlocchiniEPPS07, FlocchiniEPPS12} considered another point of view to address the same problem. The authors were interested to the problem of precomputing all the replacement minimum spanning trees when a node or an edge of the network fails. They proposed the first distributed algorithms to efficiently solve each of these problems (i.e., by considering either node or edge failure). Additional techniques and algorithms related to the construction of light weight spanning structures are extensively detailed in~\cite{Peleg00_book}.

Large scale systems are often subject to transient faults. \emph{Self-stabilization} introduced first by Dijkstra in~\cite{D74j} and later publicized by several books~\cite{Dolev00,Tel94} deals with the ability of a system to recover from catastrophic situation (i.e., the global state may be arbitrarily far from a legal state)
without external (e.g. human) intervention in finite time.

Although there already exist self-stabilizing solutions for the MST construction, none of them considered the extension of   
the Gallager, Humblet and Spira algorithm (GHS) to self-stabilizing settings.  
Interestingly, this algorithm unifies the best properties for designing large scale 
MSTs: it is fast and totally decentralized and it does not 
rely on any global parameter of the system. Our work proposes an extension of this algorithm to self-stabilizing settings. Our extension uses only poly-logarithmic memory and  
preserves all the good characteristics of the original solution in terms of 
convergence time and scalability. 

Antonoiu and Srimani, and Gupta and Srimani presented in~\cite{AntonoiuS97, GuptaS03} the first self-stabilizing algorithm for the MST problem. 
The MST construction is based on the computation of all shortest paths (for a certain cost function) between all pairs of nodes. While executing the algorithm, every node stores the cost of all paths from it to all the other nodes. To implement this algorithm, the authors assume that every node knows the number $n$ of nodes in the network, and that the identifiers of the nodes are in $\{1,\dots,n\}$. Every node $u$ stores the weight of the edge $(u,v)$ placed in the MST for each node $v\neq u$. Therefore the algorithm requires $\Omega(\sum_{v\neq u}\log w(u,v))$ bits of memory at node $u$. Since all the weights are distinct integers, the memory requirement at each node is $\Omega(n\log n)$ bits. The main drawback of this solution is its lack of scalability since each node has to know and maintain information for all the nodes in the system. Note that the authors introduce a time complexity definition related to the transmission of beacon in the context of ad-hoc networks. In a round, each node receives a beacon from all its neighbors. So, the $O(n)$ time complexity announced by the authors stays only in the particular synchronous settings. In asynchronous setting, a node is activated at the reception of a beacon from each neighbor leading to a $O(n^2)$ time complexity.
A different approach for the message-passing model was proposed by 
Higham and Liang~\cite{HighamL01}. 
The algorithm works roughly as follows: every edge checks 
whether it should belong to the MST or not. 
To this end, every non tree-edge $e$ floods the network to find 
a potential cycle, and when $e$ receives its own message back along a cycle, it uses the information collected by this message (i.e, 
the maximum edge weight of the traversed cycle) to decide whether $e$ could potentially be in the MST or not. If the edge $e$ has not 
received its message back after the time-out interval, it decides to become tree edge. The memory used by each node is $O(\log n)$ bits, but the information exchanged between neighboring nodes is of size $O(n \log n)$ bits, 
thus only slightly improving that of \cite{AntonoiuS97}. This solution
also assumes that each node has 
access to a global parameter of the system: the diameter. Its
computation 
is expensive in large scale systems and becomes even harder in dynamic settings.
The time complexity of this approach is $O(mD)$ rounds where $m$ and $D$ are the number of edges and the upper bound of the diameter of the network respectively, i.e., $O(n^3)$ rounds in the worst case.\\

In \cite{BPRT09c} we proposed a  self-stabilizing loop-free algorithm for the MST problem. 
Contrary to previous self-stabilizing MST protocols, this algorithm does not make
any assumption on the network size (including upper bounds) or the uniqueness of the
edge weights. The proposed solution improves on the memory space usage since each
participant needs only $O(\log n)$ bits while preserving the same time complexity as the algorithm in \cite{HighamL01}.  

Clearly, in the self-stabilizing implementation of the MST algorithms there is a trade-off between 
the memory complexity and their time complexity (see Table \ref{tableresume}). The challenge we address in this paper is to design fast 
and scalable self-stabilizing MST with little memory. Our approach brings together two worlds: the time efficient 
MST constructions and the memory compact informative labeling schemes. 
We do this by extending the GHS algorithm to the self-stabilizing setting while keeping it memory space compact, but using a self-stabilizing extension of the nearest common ancestor labeling scheme ~\cite{Peleg00,AGKR02}.
Note that labeling schemes have already been used in 
order to maintain compact information linked with
vertex adjacency, distance, tree ancestry or tree routing~\cite{BeinDV05},  
however none of these schemes have been studied in self-stabilizing
settings (except for the tree routing).

Our contribution is therefore twofold. 
We propose for the first time in self-stabilizing settings a $O(\log^2n)$ bits 
scheme for computing the nearest common ancestor. 
Furthermore, based on this scheme, we describe a new 
self-stabilizing algorithm for the MST
problem. Our algorithm does not make any assumption on the network size
(including upper bounds) or the existence of an a priori known root. 
The convergence time is $O(n^2)$ asynchronous rounds 
and the memory space per node is 
$O(\log^2 n)$ bits. Interestingly, our work is the first to prove the effectiveness 
of an informative labeling scheme in self-stabilizing settings and therefore 
opens a wide research path in this direction. 
The description of our algorithm is \emph{explicit}, in the sense that we describe all procedures using the formal framework $$\langle\mbox{label}\rangle : \langle\mbox{guard}\rangle \rightarrow \langle\mbox{statement}\rangle.$$

The recent paper \cite{KormanKM11} announces an improvement of our results, by sketching the \emph{implicit} description of a self-stabilizing algorithm for MST converging in $O(n)$ rounds, with a memory of $O(\log n)$ bits per node. This algorithm is also based on an informative labeling scheme. 
The approach proposed by Korman et al.~\cite{KormanKM11} is based on the composition of many sub-algorithms (some of them not stabilizing) presented in the paper as black boxes and the composition of all these modules was not proved formally correct in self-stabilizing settings up to date. The main feature of our solution in comparison with~\cite{KormanKM11} is its straightforward implementation.

\begin{table}[t]
\begin{center}
\scalebox{1}
{
\begin{tabular}{|c|c|c|c|}
\hline
 & a priori knowledge & space complexity & convergence time\\
\hline
 \cite{AntonoiuS97} &network size and & $O(n \log n)$ & $O(n^2)$ \\
 &  the nodes in the network && \\
\hline
 \cite{HighamL01}&upper bound  on diameter &$O(\log n)$& $O(n^3)$\\
 && messages of size $O(n\log n)$ & \\
\hline
\cite{BPRT09c}  &none &$O(\log n)$ &$O(n^3)$\\
\hline
This paper & none &$O(\log^2 n)$ & $O(n^2)$\\
\hline 
\end{tabular}
}
\vspace*{0,3cm}
\caption{\small Distributed Self-Stabilizing algorithms for the MST problem}
\label{tableresume}
\end{center}
\end{table}

\section{Model and overview of our solution} 
\subsection{Model}
\label{sec:model}


We consider an undirected weighted connected network $G=\langle V,E,w \rangle$ where $V$ is the set of nodes, $E$ is the set of edges and $w: E \rightarrow {\mathbb R^+}$ is a positive cost function. 
Nodes represent processors and edges represent bidirectional communication links.

The processors asynchronously execute their programs consisting of a set of variables and a finite set of rules. We consider the local shared memory model of computation\footnote{The fined-grained communication atomicity model \cite{BK07,Dolev00} can be used to design more easily a self-stabilizing algorithm for message passing model. Each node maintains a local copy of the variables of its neighbors. These variables are refreshed via special messages exchanged periodically by neighboring nodes. Therefore, in the message passing model the space complexity of our algorithm is $O(\Delta \log^2 n)$ bits per node by considering also the local copies of neighbors' variables, with $\Delta$ the maximum degree of a node in the network.}. The variables are part of the shared register which is used to communicate with the neighbors. A processor can read and write its own registers and can read the shared registers of its neighbors. 
Each processor executes a program consisting of a sequence of guarded rules. Each \emph{rule} contains a \emph{guard} (Boolean expression over the variables of a node and its neighborhood) and an \emph{action} (update of the node variables only). Any rule whose guard is \emph{true} is said to be \emph{enabled}. A node with one or more enabled rules is said to be \emph{enabled} and may execute the action corresponding to the chosen enabled rule.

A {\it local state} of a node is the value of the local variables of the node and the state of its program counter. A {\it configuration} of the system $G=(V,E)$ is the cross product of the local states of all nodes in the system. The transition from a configuration to the next one is produced by the execution of an action at a node. A {\it computation} of the system is defined as a \emph{weakly fair, maximal} sequence of configurations, $e=(c_0, c_1, \ldots c_i, \ldots)$, where each configuration $c_{i+1}$ follows from $c_i$ by the execution of a single action of at least one node. During an execution step, one or more processors execute an action and a processor may take at most one action. \emph{Weak fairness} of the sequence means that if any action in $G$ is continuously enabled along the sequence, it is eventually chosen for execution. \emph{Maximality} means that the sequence is either infinite, or it is finite and no action of $G$ is enabled in the final global state.
In this context, a \emph{round} is the smallest portion of an execution where every process has the opportunity to execute at least one action.
In the sequel we consider the system can start in any configuration. That is, the local state of a node can be corrupted. We don't make any assumption on the number of corrupted nodes. In the worst case all the nodes in the system may start in a corrupted configuration. In order to tackle these faults we use self-stabilization techniques. The definition hardly uses the legitimate predicate. A legitimate predicate is defined over the configurations of a system and describes the set of correct configurations.

\begin{definition}[self-stabilization]
Let $\mathcal{L_{A}}$ be a non-empty \emph{legitimate predicate} of an algorithm $\mathcal{A}$ with respect to a specification predicate $Spec$ such that every configuration satisfying $\mathcal{L_{A}}$ satisfies $Spec$. Algorithm $\mathcal{A}$ is \emph{self-stabilizing} with respect to $Spec$ iff the following two conditions hold:\\
\textsf{(i)} Every computation of $\mathcal{A}$ starting from a configuration satisfying $\mathcal{L_A}$ preserves $\mathcal{L_A}$ and verifies $Spec$ (\emph{closure}).  \\
\textsf{(ii)} Every computation of $\mathcal{A}$ starting from an arbitrary configuration contains a configuration that satisfies $\mathcal{L_A}$ (\emph{convergence}).
\end{definition}

To compute the time complexity, we use the definition of \emph{round}~\cite{Dolev00}. Given a computation $e$ ($e \in \mathcal{E}$), the \emph{first round} of $e$ (let us call it $e^{\prime}$) is the minimal prefix of $e$ containing the execution of one action (an action of the protocol or a disabling action) of every enabled processor from the initial configuration.  Let $e^{\prime \prime}$ be the suffix of $e$ such that $e=e^{\prime}e^{\prime \prime}$. The \emph{second round} of $e$ is the first round of $e^{\prime \prime}$.

\subsection{Overview of our solution}
\label{sec:over}
We propose to extend the Gallager, Humblet and Spira (GHS) algorithm~\cite{GallagerHS83}, to self-stabilizing settings via a compact informative labeling scheme. Thus, the resulting solution presents several advantages appealing to large scale systems: it is compact since it uses only memory whose size is poly-logarithmic  in the size of the network, it scales well since it does not rely on any global parameter of the system.

The notion of a \textit{fragment} is central to the GHS approach. A fragment is a sub-tree of the graph, i.e., a fragment is a tree which spans a subset of nodes. Note that a fragment can be limited to a single node. An outgoing edge of a fragment $F$ is an edge with a single endpoint in $F$. The minimum-weight outgoing edge of a fragment $F$ is an outgoing edge of $F$ with minimum weight among outgoing edges of $F$, denoted in the following as \MF$_F$. In the GHS construction, initially each node is a fragment. For each fragment $F$, the GHS algorithm in~\cite{GallagerHS83} identifies the \MF$_F$ and merges the two fragments endpoints of \MF$_F$. It is important to mention that with this scheme, more than two fragments may be merged concurrently. The merging process is repeated in a iterative fashion until a single fragment remains. The result is a MST. The above approach is often called \emph{blue rule} for MST construction~\cite{tar83}.

This approach is particularly appealing when transient faults create a forest of fragments (which are sub-trees of a MST). The direct application of the blue rule allows the system to reconstruct a MST and to recover from faults which have divided the existing MST. However, when more severe faults hit the system the process' states may be corrupted leading to a configuration of the network where the set of fragments are not sub-trees of some MST. This may include, a spanning tree but not a MST or spanning structure containing cycles. In these different types of spanning structures, the application of the \emph{blue rule} is not always sufficient to reconstruct a MST. To overcome this difficulty, we combine the \emph{blue rule} with another method, referred in the literature as the \emph{red rule}~\cite{tar83}. The \emph{red rule} considers all the possible cycles in a graph, and removes the heaviest edge from every cycle, the resulting is a MST. To maintain a MST regardless of the starting configuration, we use the \emph{red rule} as follows. Let ${\tt T}$ denote a spanning tree  of graph $G$, and $e$ an edge in $G$ but not in ${\tt T}$. Clearly, if $e$ is added to ${\tt T}$, this creates a (unique) cycle composed by $e$ and some edges of ${\tt T}$. This cycle is called a \textit{fundamental cycle}, and denoted by $C_e$. 
According to the \emph{red rule}, if $e$ is not the edge of maximum weight in $C_e$, then there exists an edge $f \neq e$ in $C_e$, $f\in T$ such that $w(f)>w(e)$. In this case, $f$ can be removed since it is not part of any MST. 

Our solution, called in the following \MSTA\/ Algorithm, combines both the \emph{blue rule} and \emph{red rule}. The application of the \emph{blue rule} needs that each node identifies the fragment it belongs to. The \emph{red rule} also requires that each node can identify the fundamental cycle associated to each of its adjacent non-tree-edges. Note that a simple scheme broadcasting the root identifier in each fragment (of memory size $O(\log n)$ bits per node) can be used to identify the fragments, but this cannot allow to identify fundamental cycles. In order to identify fragments or fundamental cycles, we use a self-stabilizing labeling scheme, called $\LabA$. This scheme provides at each node a distinct label. For two nodes $u$ and $v$ in the same fragment, the comparison of their labels provides to these two nodes their \textbf{nearest common ancestor} in a tree (see Section~\ref{sec:label}). Thus, the advantage of this labeling is twofold. First the labeling scheme helps each node to identify the fragment it belongs to. Second, given any non-tree edge $e=\{u,v\}$, the path in the tree going from $u$ to the nearest common ancestor of $u$ and $v$, then from there to $v$, and finally back to $u$ by traversing $e$, constitute the fundamental cycle $C_e$.

To summarize, \MSTA\/ algorithm will use the \emph{blue rule} to build a spanning tree, and the \emph{red rule} to recover from invalid configurations. In both cases, it uses our algorithm \LabA\/ 
to identify both fragments and fundamental cycles. Note that, in~\cite{ParkMHT90,ParkMHT92} distributed algorithms using the \emph{blue} and \emph{red rules} to construct a MST in a dynamic network are proposed, however these algorithms are not self-stabilizing.

\subsection{Notations}

In this section we fix some general assumptions used in the current paper. Let $G=\langle V,E,w \rangle$ be an undirected weighted graph, where $V$ is the set of nodes, $E$ is the set of edges and the weight of each edge is given by a positive cost function $w: E \rightarrow {\mathbb R^+}$. We consider w.l.o.g. that the edges' weight are polynomial in $|V|$. Moreover, the nodes are allowed to have unique identifiers denoted by $\id$ encoded using $O(\log n)$ bits where $n=|V|$. No assumption is made about the fact that edges' weight must be distinct. In the current paper  $\N_v$ denotes the set of all neighbors of $v$ in $G$, for any node $v \in V(G)$.

Each node $v$ maintains several information a pointer to one of its neighbor node called \textit{the parent}. The set of these pointers induces a spanning tree if the spanning structure is composed with all the nodes and contains no cycle. We denote by $\path(u,v)$ the path from $u$ to $v$ in the tree. For handling the nearest common ancestor labeling scheme we will define some notations. Let $\lab_v$ be the label of a node $v$ composed by a list of pairs of integers, where each pair is an identifier and a distance. $\lab_v[i]$ denotes the $i$th pair of the list, and for every pair $i$ the first element is denoted by $\lab_v[i][0]$ and the second one by $\lab_v[i][1]$. The last pair of the list is denoted by $\lab^{-1}_v$.

%

\section{Self-stabilizing Nearest Common Ancestor Labeling scheme}
\label{sec:label}

Previously, we explained that our  \MSTA\/ algorithm needs to identify fragments, internal and outgoing edges of each fragment and the presence of cycles. To achieve this identification we use a nearest common ancestor labeling scheme. This section is dedicated to the presentation of a self-stabilizing version of the distributed algorithm proposed by Peleg~\cite{Peleg00}. The self-stabilizing algorithm is called in the following  \LabA.  \LabA\/ algorithm can be used to solve other tasks than constructing a MST, hence we present this part in a separate section.

In~\cite{Peleg00}, Peleg gives a nearest common ancestor labeling scheme for a tree structure with a memory complexity of $\Theta(\log^2n)$ bits. We will first give in this section a self-stabilizing version of this scheme, that is the encoder and decoder part related to the labeling scheme, and finally we prove the correctness and the complexity of our self-stabilizing algorithm. For simplicity, we assume in this current section that all the nodes of the network belong to a single tree. It is easy to see that without a tree structure, the nodes cannot have a common ancestor. Therefore, in the next section we have to deal with the general case in which cycles can be contained in the starting configuration.

\subsection{Variables}

Before presenting the nearest common ancestor labeling scheme, we describe below the variables used by the labeling scheme. Each node $v \in V$ maintains three variables:
\begin{itemize}
\item A parent pointer to a neighbor of $v$ stored in $\parent_v$ defining the spanning tree.
\item $\size_v$ is a pair of integers, whose the first element is an estimation of the number of nodes in the sub-tree of $v$ and the second one is the identifier of the child of $v$ with the subtree of highest size. If $v$ has no child then $\size_v$ is equal to $(1,\bot)$. Note that, the first integer of the pair is referenced by $\size_v[0]$, while the second integer by $\size_v[1]$.
\item The label of $v$ (composed of a list of pairs of integers where each pair is an identifier and a distance (described below)) is stored in variable $\lab_v$.
\end{itemize}

We will now present the manner the nearest common ancestor scheme computes the label of each node in a spanning tree.

\subsection{Labeling encoder}
The main idea of this protocol is to divide a tree structure in sub-paths to minimize the label size of each node. Let us describe more precisely our self-stabilizing version of this protocol. In a rooted tree, a \textit{heavy} edge is an edge between a node $u$ and one of its children $v$ with the highest number of nodes 
in its sub-tree. The other edges between $u$ and its other children are tagged as \textit{light} edges. We extend this edge designation to the nodes, a node $v$ is called \emph{heavy node} if the edge between $v$ and its parent is a heavy edge (see Predicate $\Heavy(v)$ in Figure~\ref{fig:predicates1}), otherwise $v$ is called \emph{light node} (see Predicate $\Light(v)$ in Figure~\ref{fig:predicates1}). Moreover, the root of a tree is a heavy node. The idea of the scheme is as follows. A tree is recursively divided into paths of disjoint edges: \emph{heavy} and \emph{light} paths. Remark: Any child of highest number of nodes can be selected as heavy node, so among these children the one of highest identifier can be selected.

\begin{figure}
\fbox{
\begin{minipage}{16cm}
\begin{small}
\begin{tabular}{lll}
$\Child(v)$ & $=$ & $\{u \in \N_v : \parent_u=\id_v\}$\\
$\NbrNdS(v)$ & $=$ & $\Big(1+\mathlarger{\sum}_{u \in \Child(v)} \size_u[0], \quad \max\{\id_u: u \in \Child(v) \wedge \size_u[0]=\max\{\size_x[0]: x \in \Child(v)\}\}\Big)$\\
\end{tabular}\\

\begin{tabular}{lll}
$\Isleaf(v)$ & $\equiv$ & $(\Child(v)=\emptyset  \wedge \size_v=(1,\bot))$\\
$\SizeC(v)$ & $\equiv$ & $\Isleaf(v) \vee (\Child(v)\neq\emptyset  \wedge \size_v=\NbrNdS(v))$\\
$\LabelR(v)$ & $\equiv$ & $(\parent_v=\emptyset \wedge \lab_v=(\id_v,0))$\\
$\LabelNd(v)$ & $\equiv$ & $(\parent_v \in N(v)) \wedge (\Heavy(v) \vee \Light(v))$\\
$\Label(v)$ & $\equiv$ & $\LabelR(v) \vee \LabelNd(v)$\\
$\Heavy(v)$ & $\equiv$ & $(\size_{\parent_v}[1]=\id_v )\wedge (\size_v[0]< \size_{\parent_v}[0]) \wedge (\lab_{\parent_v}\backslash\lab^{-1}_{\parent_v}=\lab_v\backslash\lab^{-1}_v) \wedge (\lab^{-1}_{\parent_v}[1]+1=\lab^{-1}_v[1])$\\
$\Light(v)$ & $\equiv$ & $(\size_{\parent_v}[1]\neq \id_v)  \wedge (\size_v[0]\leq \size_{\parent_v}[0]/2) \wedge (\lab_v=\lab_{\parent_v}.(\id_v,0))$\\
\end{tabular}\\
\end{small}
\end{minipage}
}
\caption{Macros and predicates of Algorithm \LabA\/ for any $v \in V$.}
\label{fig:predicates1}
\end{figure}

%

To label the nodes in a tree $\T$, the size of each subtree rooted at each node of $\T$ is needed to identify heavy edges leading the heaviest subtrees at each level of $\T$. To this end, each node $v$ maintains a variable named $\size_v$ which is a pair of integers. The first integer is the local estimation of the number of nodes in the subtree rooted at $v$. For a node $v$ this value is computed by summing up all the estimated values of its children plus one. The value of $\size_v$ is processed in a bottom-up fashion from the leaves to the root of the tree (see  Predicate $\SizeC(v)$ in Figure~\ref{fig:predicates1} and rule $\RSC$). The second integer is the identifier of a child of $v$ with maximum number of nodes in its sub-tree, which indicates the heavy edge. We suppose w.l.o.g that, in case of equality between the size of the children's subtrees the child with the minimum identity is chosen. The variable $\size_v$ is setted to $(1,\bot)$ for a leaf node $v$ (see Predicate $\Isleaf(v)$ in Figure~\ref{fig:predicates1}).

Based on the heavy and light nodes in a tree $\T$ indicated by variable $\size_v$ at each node $v \in \T$, each node of $\T$ can compute its label (see rule $\RLC$ in Figure~\ref{fig:algo_lab}). The label of a node $v$ stored in $\lab_v$ is a list of pair of integers. Each pair of the list contains the identifier of the node which is the root of the heavy path (i.e., a path including only heavy edges) that $v$ belongs to and the distance to it. For the root $v$ of a fragment, the label $\lab_v$ is the following pair $(\id_v,0)$, respectively the identifier of $v$ and the distance to itself, i.e., zero (see Predicate $\LabelR(v)$ in Figure~\ref{fig:predicates1}). When a node $u$ is tagged by its parent as a heavy node (i.e., $\size_{\parent_v}[1]=\id_u$), then the node $u$ takes the label of its parent but it increases by one the distance of the last pair of the parent label (see Predicate $\LabelNd(v)$ in Figure~\ref{fig:predicates1}).\\
Otherwise, a node $u$ is tagged by its parent $v$ as a light node (i.e., $\size_{\parent_v}[1] \neq \id_u$), then the node $u$ becomes the root of a heavy path and it takes the following label: the label of its parent to which $u$ concatenates to a new pair composed by its identifier and a zero distance (we note the step of concatenation by the operator ".").\\
Examples of theses cases are given in Figure~\ref{fig:label}, where integers inside the nodes are node identifiers and lists of pairs of values are node labels.

\begin{figure}[tbp]
\begin{center}
\includegraphics[scale=0.65]{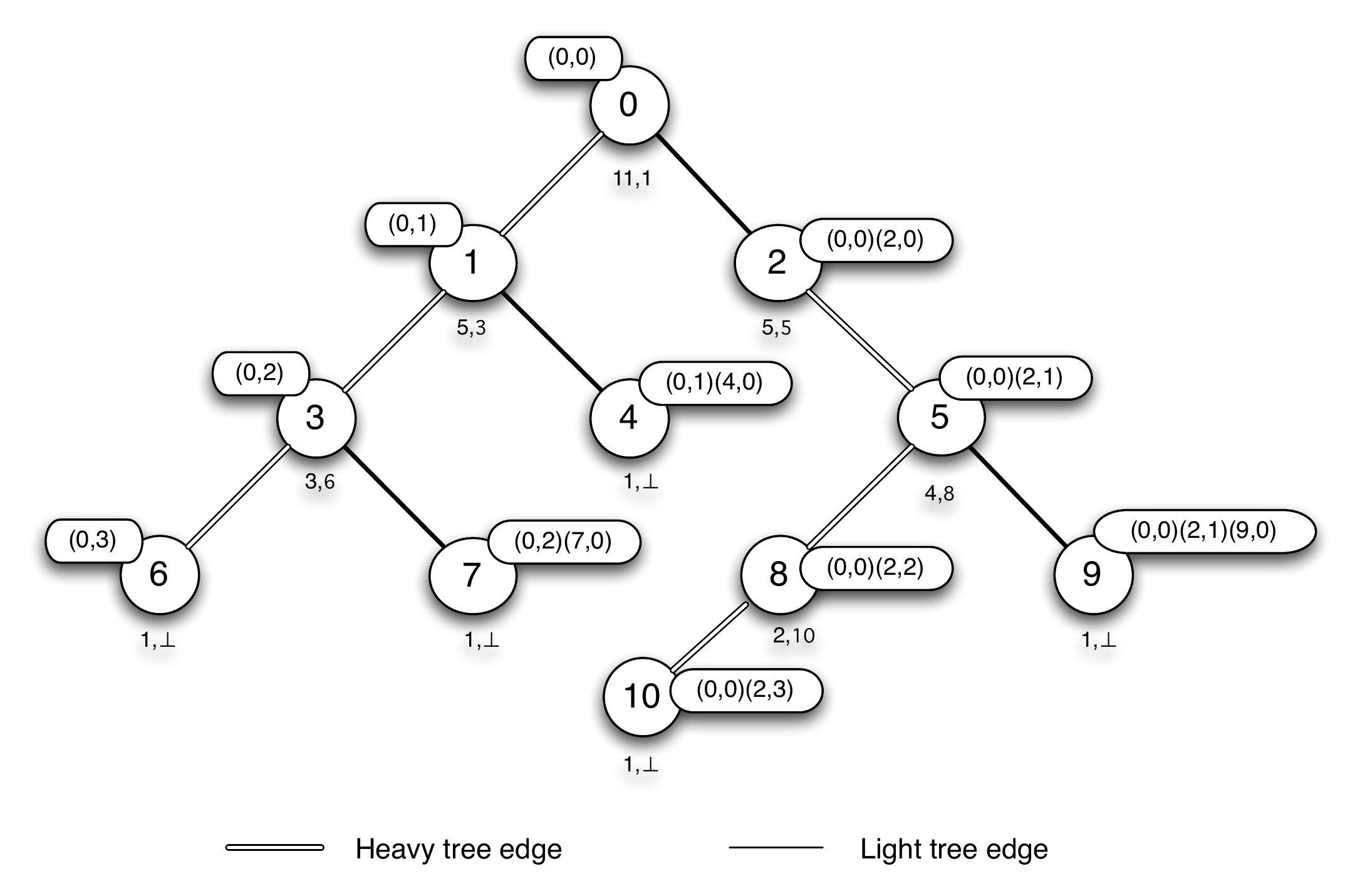}
\caption{Nearest Common Ancestor Labeling scheme for a tree. The bubble at each node $v$ corresponds to the label of $v$. The integer inside each node corresponds to the node's identifier, while the other notation corresponds to the variable $\size$.}
\label{fig:label}
\end{center}
\end{figure}

Algorithm \LabA\/ is composed by the rules $\RSC$ and $\RLC$ given in Figure~\ref{fig:algo_lab} which correct the variables \size\/ and \lab\/ respectively if needed.

\begin{figure}
\fbox{
\begin{minipage}{15cm}
\begin{description}
\item[$\RSC$: [\ Size correction]]
\item\textbf{If} $\neg \SizeC(v)$ \textbf{Then}\\
\hspace*{0,1cm}\textbf{If} $\Child(v)=\emptyset$ \textbf{then}  $\size_v:=(1,\bot)$\\
\hspace*{0,1cm}\textbf{Else} $\size_v:= \NbrNdS(v);$
\item[$\RLC$: [\ Label correction]]
\item\textbf{If} $\SizeC(v) \wedge \neg \Label(v)$ \textbf{Then}\\
\hspace*{0,1cm}\textbf{If} $\size_{\parent_v}[1]=\id_v$ \mbox{\textbf{then }} $\lab_v:=\lab_{\parent_v};\lab^{-1}_v[1]:=\lab^{-1}_v[1]+1;$\\
\hspace*{0,1cm}\textbf{Else} $\lab_v:=\lab_{\parent_v}.(\id_v,0)$
\end{description}
\end{minipage}
}
\caption{Formal description of Algorithm \LabA\/ for any $v \in V$.}
\label{fig:algo_lab}
\end{figure}

\subsection{Labeling decoder}
\label{subsec:decoder}

Let us now describe the decoder for the nearest common ancestor. This decoder is given in~\cite{Peleg00}, but for simplicity we present it using our own notations (see predicate $\Lca$ in Figure~\ref{fig:nca}). Let us consider two nodes $u$ and $v$, we denote by $\Lca(\lab_u,\lab_v)$ the label of the nearest common ancestor of $u$ and $v$. For the remainder of this paper, we define the following notations: $\lab^{\mathsmaller{\cap}}_{u,v}=\lab_u\cap\lab_v$ and $\lab'_{u,v}=\lab_u\backslash\lab^{\mathsmaller{\cap}}_{u,v}$. The nearest common ancestor of $u$ and $v$ is composed by the common part of the label of $u$ and $v$ ($\lab^{\mathsmaller{\cap}}_{u,v}$) and by the smaller pair following the lexicographic order of the last pair of their labels (i.e., minimum between $\lab'_u[0]$ and $\lab'_v[0]$). In the other case $u$ and $v$ have not common ancestor.

\begin{figure}[ht]
\fbox{
\begin{minipage}{15cm}
\begin{description}
\item $\Lca(\lab_u,\lab_v) \equiv \left\{ \begin{array}{ll}
\lab^{\mathsmaller{\cap}}_{u,v}.\lab'_{u,v}[0]&\hspace{0,5cm}\mbox{\textbf{If} $\lab'_{u,v}[0][0]=\lab'_{v,u}[0][0] \wedge \lab'_{u,v}[0][1]<\lab'_{v,u}[0][1] \wedge \lab^{\mathsmaller{\cap}}_{u,v} \neq \emptyset$}\\
\lab^{\mathsmaller{\cap}}_{u,v}.\lab'_{v,u}[0]&\hspace{0,5cm}\mbox{\textbf{If} $\lab'_{u,v}[0][0]=\lab'_{v,u}[0][0] \wedge \lab'_{v,u}[0][1]>\lab'_{u,v}[0][1] \wedge \lab^{\mathsmaller{\cap}}_{u,v} \neq \emptyset$}\\
\emptyset&\hspace{0,5cm}\mbox{\textbf{otherwise}}\\
\end{array} \right.$\\
\end{description}
\end{minipage}
}
\caption{Macro used for computing the nearest common ancestor.}
\label{fig:nca}
\end{figure}

On the example defined on  Figure~\ref{fig:label}, the labels of nodes $9$ and $10$ are respectively $\lab_{9}=(0,0)(2,1)(9,0)$ and $\lab_{10}=(0,0)(2,3)$. In this case, we have for the defined notations on labels: $\lab^{\mathsmaller{\cap}}_{9,10}=(0,0)$, $\lab'_{10,9}=(2,1)(9,0)$ and $\lab'_{9,10}=(2,3)$. Since we have $\lab'_{9,10}[0][1]<\lab'_{10,9}[0][1]$ then on this example $\Lca(\lab_9,\lab_{10})=(0,0)(2,1)$.

\subsection{Correctness and complexity}
\label{subsec:CorLabel}
This subsection is dedicated to the correctness of the self-stabilizing nearest common ancestor labeling scheme. Let $\Gamma$ be the set of all possible configurations of the system. In order to prove the correctness of the $\LabA\/$ algorithm, we denote $\Gamma_{\size}$ the set of configurations in $\Gamma$ such that variables $\size$ are correct in the system. More precisely, we define the following   function $s$: $V\rightarrow \mathbb{N}$ be the function defined by $$s(v)=|\big((\size_v[0]-1)-\mathlarger{\sum}_{\mathsmaller{u\in \mathcal{C}(v)}} \size_u[0]\big)|.$$
Note that $s(v)\geq 0$, and the variable $\size$ has a correct value at node $v$ if and only if $s(v)=0$.
In the following, we show that any execution of the system converges to a configuration in $\Gamma_{\size}$, and the set of configurations $\Gamma_{\size}$ is closed. The following lemma establishes the former property. We assume that all the nodes of the system belongs to the tree ${\tt T}$ and we define below a legitimate configuration for the informative labeling scheme considered in this section.

\begin{definition}[Legitimate configuration for Labeling scheme]
A configuration $\gamma \in \Gamma_\Lambda$ is called legitimate if the following conditions are satisfied:
\begin{enumerate}
\item the root node $r$ of the tree ${\tt T}$ has label equal to $(\id_r, 0)$,
\item every \emph{heavy} node $v \in {\tt T}$ has a label equal to $(\lab_{\parent_v} \backslash \lab^{-1}_{\parent_v}).(\lab^{-1}_{\parent_v}[0], \lab^{-1}_{\parent_v}[1]+1)$,
\item every \emph{light} node $v \in {\tt T}$ has a label equal to $\lab_{\parent_v}.(\id_v, 0)$.
\end{enumerate}
\end{definition}

\begin{lemma}
Starting from an arbitrary configuration $\gamma \in \Gamma$, the system reaches a configuration in $\gamma' \in \Gamma_{\size}$  in $O(\delta_{\tt T})$ rounds, where $\delta_{\tt T}$ is  the depth of the tree ${\tt T}$.
\label{lem:size_convergence}
\end{lemma}
\begin{proof}
First, we  define the following  potential function $\Phi$. 
We denote by~$\delta_{\tt T}$ the depth of the tree ${\tt T}$, i.e., the length of the longest path from the root to the leaves.
Let $\gamma \in \Gamma$ be a configuration, and let $\Gamma$ be the set of all configurations. Let $\Phi: \Gamma\rightarrow\mathbb{N}$ be the function defined by $$\Phi(\gamma)=\mathlarger{\sum}_{\mathsmaller {d=0}}^{\mathsmaller{\delta_T}} \nu_d(\gamma) (n+1)^d$$
where $\nu_d(\gamma)$ is the number of nodes $v$ at depth $d$ in ${\tt T}$ with $s(v)\neq 0$.  
Note that $0\leq \nu_d(\gamma) \leq n$, and $0\leq \Phi(\gamma)\leq (n+1)^{\delta_T+1}$. Also, the variable $\size$ has a correct value at every node if and only if  $\Phi(\gamma)=0$. Let $\gamma(t)$ denotes the configuration of the system after round $t$. Let $d_0$ be the largest index such that $\nu_{d_0}(\gamma(t))\neq 0$. Since we use a weakly fair scheduler, all the nodes are scheduled during the execution of round $t+1$. Every node $v$ at depth $d>d_0$ does not change its value of variable $\size$ (see the predicate $\SizeC$), and therefore $s(v)$ remains zero, so $\nu_d(\gamma(t+1))$ remains zero as well. The nodes at depth $d_0$ change their variable $\size$ according to the variable $\size$ of their children. Let $v$ be a node at depth $d_0$.  The children of $v$ (if any) are at depth $d>d_0$. Thus, their variable $\size$ has not changed, and therefore $s(v)$ becomes zero after round $t+1$. As a consequence, $\nu_{d_0}(\gamma(t+1))=0$. Therefore, we get 
$$\Phi(\gamma(t+1))<\Phi(\gamma(t))$$
and  thus the system will eventually reach a configuration in $\Gamma_{\size}$. To measure the number of rounds it takes to get into $\Gamma_{size}$, observe that $\delta_{\tt T}$ decreases by at least one at each round.Starting from any arbitrary configuration, the system reaches a configuration in $\Gamma_{\size}$ in $O(\delta_{\tt T})$ rounds.
\end{proof}

\begin{lemma}
Starting from a configuration in $\Gamma_{\size}$ the system can only reach configurations in $\Gamma_{\size}$.
\label{lem:size_closure}
\end{lemma}
\begin{proof}
According to algorithm $\LabA$, the variable $\size$ is modified only by Rule $\RSC$. Consider a configuration $\gamma \in \Gamma_{\size}$ such that variables $\size$ are correct. For each node $v$, we have $s(v)=0$ and Predicate $\SizeC(v)$ is true. Thus, Rule $\RSC$ cannot be executed by a node $v$ and we have $s(v)=0$ which implied that $\Phi(\gamma)=0 $. Therefore, for any execution starting from a configuration $\gamma \in \Gamma_{\size}$, the system remains in a configuration in $\Gamma_{\size}$.
\end{proof}

\begin{lemma}[Convergence for $\LabA$]
 \label{lem:label_convergence}
 Starting from an illegitimate configuration, Algorithm \LabA\/  reaches in $O(\delta_{\tt T})$ rounds a legitimate configuration, where $\delta_{\tt T}$ is  the depth of the tree ${\tt T}$.
 \end{lemma}
\begin{proof}
Let us introduce some notations that we will use throughout in the proof. Let \mbox{ $\bar \lab_v=\lab_v\backslash \lab^{-1}_v$}  be the pairs list of the node's label $v$ such that the last pair is removed, and $|\lab_v|$ the number of pairs in the label of $v$. For two labels $\lab_v$ and $\lab_u$ the step $\circleddash$ is defined by: $$\lab_v\circleddash \lab_u=\mathlarger{\sum}_{i=0}^{|\lab_v|-1}|\lab_v[i][0]-\lab_u[i][0]| + |\lab_v[i][1]-\lab_u[i][1]|.$$

We first define a first function $L(v)$ on the state of each node $v \in V$ as following:
\begin{small}
$$L(v) \equiv \left\{ \begin{array}{ll} s(v)+||\lab_v|-1|+\lab_v \circleddash (\id_v, 0) & \mbox{If } v=r\\
s(v)+||\lab_v|-|\lab_{\parent_v}||+(\bar \lab_v\circleddash\bar\lab_{\parent_v})+|\lab^{-1}_v[0]-\lab^{-1}_{\parent_v}[0]| +|\lab^{-1}_v[1]-\lab^{-1}_{\parent_v}[1]-1| &\mbox{If } \size_{\parent_v}[1]=\id_v\\
s(v)+||\lab_v|-|\lab_{\parent_v}|-1|+(\bar \lab_v\circleddash\lab_{\parent_v})+(\lab^{-1}_v\circleddash(\id_v,0)) & \mbox{Otherwise}
\end{array} \right.$$
\end{small}

Note that $L(v)\geq 0$ and when $L(v)=0$  the variable $\lab$ has a correct value for a node $v$. Let $\Lambda$: $\Gamma\rightarrow\mathbb{N}$ be the function defined by, $$\Lambda(\gamma)=\mathlarger{\sum}_{\mathsmaller {d=0}}^{\mathsmaller{\delta_T}} \xi_d(\gamma) (n+1)^{n+1-d}$$ 
where $\xi_d(\gamma)$ is the number of nodes $v$ at depth $d$ in ${\tt T}$ with $L(v)\neq 0$.  
Remark that $0\leq \xi_d(\gamma) \leq n$, and $0\leq \Lambda(\gamma)$. Also, the variable $\lab$ has a correct value at every node if and only if  $\Lambda(\gamma)=0$. 
 Let $\gamma(t)$ denote the configuration of the system after round $t$, and suppose that $t>n$.   
By lemma~\ref{lem:size_convergence} and lemma~\ref{lem:size_closure} we prove that $\gamma(t) \in \Gamma_{\size}$. 
It is important to mention that, in $\gamma(t)$ all node $v$ can check if it is a heavy node or light node (see variable $\size$).   Let $d_0$ be the smallest index such that $\xi_{d_0}(\gamma(t))\neq 0$. Since we use a weakly fair scheduler, all the nodes are scheduled during the execution of round $t+1$. Every node $v$ at depth $d<d_0$ does not change its value of variable $\lab$ (see the predicate $\Label$), and therefore $L(v)$ remains zero, so $\xi_d(\gamma(t+1))$ remains zero as well. The nodes at depth $d_0$ change their variable $\lab$ according to the variables $\lab$ and $\size[1]$ of their parent (see Rule $\RLC$). Let $v$ be a node at depth $d_0$.  The parent of $v$  is at depth $d<d_0$. 
Thus, its variable $\lab$ have not changed, and therefore $L(v)$ becomes zero after round $t+1$. As a consequence, $\xi_{d_0}(\gamma(t+1))=0$. Therefore, we get 
$$\Lambda(\gamma(t+1))<\Lambda(\gamma(t))$$
and  thus the system will eventually reach a legitimate configuration for algorithm $\LabA$. To measure the number of rounds it takes to get into a legitimate configuration  for algorithm $\LabA$, observe that $\delta_T$ decreases by at least one at each round. Since $\delta_T \leq n-1$ for every $\gamma\in\Gamma_{\size}$, we get that, starting from any configuration in $\Gamma_{\size}$ configuration, the system reaches a legitimate configuration  for algorithm $\LabA$ in $O(\delta_{\tt T})$ rounds. Using lemma~\ref{lem:size_convergence} and lemma~\ref{lem:size_closure}, we can conclude starting from any arbitrary configuration, the system reaches a legitimate configuration for algorithm $\LabA$ in $O(\delta_{\tt T})$ rounds.
 \end{proof}

\begin{lemma}[Closure for $\LabA$]
 \label{lem:label_closure}
 The set of legitimate configurations for \LabA\/ is closed. That is, starting from any legitimate configuration, the system remains in a legitimate configuration. 
 \end{lemma}

 \begin{proof}
 According to Algorithm $\LabA$, the labeling procedure is done using only Rule $\RLC$. Let $\gamma$ a legitimate configuration. For each node $v$ in $\gamma$, we have $\Phi(\gamma)=0$ and $\Lambda(\gamma)=0$. Moreover in $\gamma$, Predicates $\SizeC(v)$ and $\Label(v)$ are true and Rules $\RLC$ and $\RSC$ cannot be executed by any node $v \in V$. In conclusion, starting from a legitimate configuration for algorithm $\LabA$ the system remains in a legitimate configuration.
 \end{proof}

The following theorem is a direct consequence from Lemmas~\ref{lem:label_convergence} and \ref{lem:label_closure}.

 \begin{theorem}
\label{thm:self-stab_lab}
 Algorithm $\LabA$ is self-stabilizing for the informative nearest common ancestor labeling scheme.
 \end{theorem}

\section{Self-Stabilizing Minimum Spanning Tree Algorithm}
\begin{figure}[htbp]
\begin{center}
\includegraphics[scale=0.6]{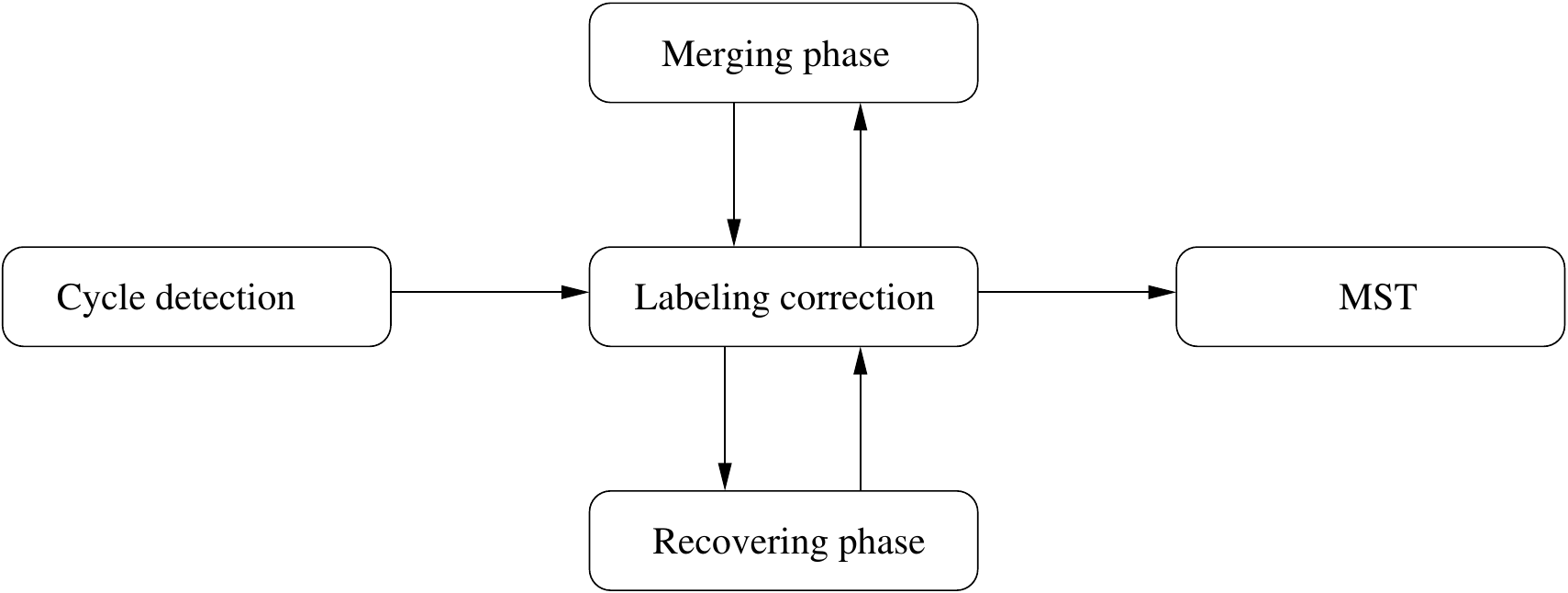}
\caption{Relation between the different phases of Algorithm $\MSTA$.}
\label{fig:schemaRules}
\end{center}
\end{figure}

In this section we describe our self-stabilizing algorithm for constructing the minimum spanning tree, called \MSTA\/ algorithm.
Our  \MSTA\/ algorithm uses the ``blue rule" to construct a spanning tree and the ``red rule" to recover from invalid configurations (see section \ref{sec:over}). 
In both cases, it uses   \LabA\/ algorithm to identify fragments and fundamental cycles. We assume in the following that the \textit{merging} phases have a higher priority than the \textit{recovering} phases. That is, the system recovers from an invalid configuration if and only if no merging is possible. 

Unfortunately, due to arbitrary initial configuration, the structure induced by the parent pointer of all nodes may contain cycles. We use first a well known approach to break cycles before giving a detailed description of merging and recovering phases.

Figure~\ref{fig:schemaRules} illustrates the different phases of Algorithm $\MSTA$. Starting  from an arbitrary configuration, first all the cycles are destroyed then fragments are defined and correctly labeled using the parent pointers. Based on the label of nodes, the minimum \emph{outgoing edge} (i.e., edge whose extremities belong to different fragments) of each fragment is computed in a bottom-up fashion, and allowing to a pair of fragments which have selected the same outgoing edge to be merged together through this edge. A \emph{merging step} gives a new fragment which is the result of the merging of a pair of fragments. When a new fragment is created, the nodes of this fragment have to compute their new label. This process is repeated until there is only one remaining fragment spanning all the nodes of the network. In this case, the recovering phase can begin by detecting that no outgoing edge can be selected. To handle this phase each fragment has to compute its \emph{internal edges} (i.e., edges whose extremities belong to the same fragment) and to identify the \emph{nearest common ancestor} based on the labels of the edge extremities. The weight of the internal edges are broadcasted up in the tree from the leaves to the root. Let $e=\{u,v\}$ an internal edge of Tree ${\tt T}$, due to the ``red rule" if an edge $f$ of the path $\path(u,\Lca(\lab_u,\lab_v))$ in ${\tt T}$ has a weight bigger than $e$, then $e$ is an \emph{valid edge} since $e$ is part of an ${\tt MST}$ (by ``red rule''). More precisely, if during the bottom-up transmission of the weight of $e$, a node $u$ has a parent link edge $f$ such that $w(f)>w(e)$ then $f$ is deleted from the tree ${\tt T}$ and $u$ becomes the root of a new fragment.

We present first the variables used by Algorithm $\MSTA$, then we describe the approach used to delete the cycles, followed by the merging and recovering phases. Finally, we show the correctness and the time and memory complexities of the algorithm.

\subsection{Variables}

We list below the eight variables maintained at each node $v \in V$:
\begin{itemize}
\item The three variables described in Section~\ref{sec:label} are used, i.e., variables $\parent_v, \size_v$ and $\lab_v$.
\item The distance of each node $v$ from the root of the fragment is stored in variable $d_v$.
\item For handling the \emph{blue rule} mentioned in section~\ref{sec:over}, the minimum outgoing edge of each fragment is stored in Variable $\m_v$. This edge is composed of three elements: the edge weight, and the identifiers of the edge extremities. The $i$-th element of $\m_v$ is accessed by $\m_v[i]$ with $i\in\{0,...,2\}$.
\item Finally to broadcast the internal edges in the recovering phase, a last variable $\mIns_v$ stores three elements related to an internal edge: the edge weight, and the labels of the edge extremities. As for Variable $\m_v$, the $i$-th element of $\mIns_v$ is accessed by $\mIns_v[i]$ with $i\in\{0,...,2\}$.
\end{itemize}

\subsection{Cycles detection and Labels correction}
The previous section was dedicated to the labeling procedure for an unique tree, due to the arbitrary starting configuration, the network can contain a forest of subtrees (several fragments) and cycles. Therefore, the labeling procedure described in previous section (using Rules $\RSC$ and $\RLC$) is executed separately in each subtree in Algorithm $\MSTA$. However, to apply this procedure it is crucial to detect the cycles in the fragments induced by the parent pointers. To this end, we use a common approach used to break cycles in a spanning structure~\cite{DolevIM97}. Each node computes its distance (in hops) to the root by using the distance of its parent plus one. By following this procedure, there is at least a node which has a distance higher or equal than the distance of its parent in the fragment. Therefore, this condition is used at each node to detect a cycle. In this case, a node $v$ deletes its parent pointer by selecting no parent and a new fragment rooted at $v$ is created. Unfortunately, due to the arbitrary initial configuration a cycle can be falsely detected because of erroneous distances values at $v$ and its parent. This mechanism based on distances ensures that after $O(n)$ rounds the network is cycle free. The destruction of cycles is managed by rule $\RCo$. 

When all the cycles have been deleted, the labeling procedure is applied in Algorithm $\MSTA$. Note that the cycle detection must have a higher priority over the labeling procedure. To this end, Rule $\RCo$ is the first rule to execute and in exclusion with Rules $\RSC$ and $\RLC$ in Algorithm $\MSTA$. Furthermore, the labeling scheme must also have a higher priority over the merging and recovering phases. Indeed, the label of the nodes are used to identify the internal and outgoing edges of a fragment (see Figure~\ref{fig:fusion}). To guarantee the execution priority, the rules of the labeling scheme can only be executed when Predicate $Distance(v)$ is satisfied at node $v$. In the same way, the rules of merging and recovering phases can only be executed at a node $v$ when Predicate $CorrectF(v)$ is satisfied at $v$.

\begin{figure}[!ht]
\fbox{
\begin{minipage}{15cm}
\begin{tabular}{lll}
$\Distance(v)$ & $\equiv$ & $(p_v=\emptyset \wedge d_v=0) \vee (p_v\neq \emptyset \wedge d_v=d_{p_v}+1)$\\
$\SizeC(v)$ & $\equiv$ & $\Isleaf(v) \vee (\Child(v)\neq\emptyset  \wedge \size_v=\NbrNdS(v))$\\
$\Label(v)$ & $\equiv$ & $\LabelR(v) \vee \LabelNd(v)$\\
$CorrectF(v)$ & $\equiv$ & $\Distance(v) \wedge \SizeC(v) \wedge \Label(v)$\\
\end{tabular}
\end{minipage}
}
\caption{Predicates used by Rule $\RCo$ and labeling rules.}
\label{fig:predicatesRCo}
\end{figure}

\bigskip

%

\begin{minipage}{16cm}
\hrulefill
\begin{description}
\item[$\RCo$: [\ Correction]]
\item \textbf{If} $\neg \Distance(v)$ \textbf{Then}\\
\hspace{1,5cm}$\m_v=\emptyset;\mIns_v=\emptyset$\\
\hspace{1,5cm}\textbf{If} $(p_v=\emptyset) \wedge d_v\neq 0$ \textbf{Then} $d_v:=0;$\\
\hspace{1,5cm}\textbf{If} $(p_v\neq\emptyset) \wedge (d_{p_v}+1<d_v)$ \textbf{Then} $d_v:=d_{p_v}+1;$\\
\hspace{1,5cm}\textbf{If} $(p_v\neq\emptyset) \wedge (d_{p_v} \geq d_v)$ \textbf{Then} $\parent_v:=\emptyset;\ \lab_v:=(\id_v,0); d_v:=0;$
\end{description}
\hrulefill
\vspace*{0.5cm}
\end{minipage}

We give below the rules associated with the labeling encoder (given in the previous section). In order to use these two rules for the MST construction, we add Predicate $\Distance(v)$ in the guards. This allow to disable these rules when a cycle is detected with Rule $\RCo$.

\begin{minipage}{16cm}
\vspace*{0.5cm}
\hrulefill
\begin{description}
\item[$\RSC$: [\ Size correction]]
\item\textbf{If} $\Distance(v) \wedge \neg \SizeC(v)$ \textbf{Then}\\
\hspace*{0,1cm}\textbf{If} $\Child(v)=\emptyset$ \textbf{then}  $\size_v:=(1,\bot)$\\
\hspace*{0,1cm}\textbf{Else} $\size_v:= \NbrNdS(v);$
\item[$\RLC$: [\ Label correction]]
\item\textbf{If} $\Distance(v) \wedge \SizeC(v) \wedge \neg \Label(v)$ \textbf{Then}\\
\hspace*{0,1cm}\textbf{If} $\size_{\parent_v}[1]=\id_v$ \mbox{\textbf{then }} $\lab_v:=\lab_{\parent_v};\lab^{-1}_v[1]:=\lab^{-1}_v[1]+1;$\\
\hspace*{0,1cm}\textbf{Else} $\lab_v:=\lab_{\parent_v}.(\id_v,0)$
\end{description}
\hrulefill
\end{minipage}

\subsection{Merging phase}

When the graph induced by the parent pointers is cycle free and every node $v$ of a fragment $F$ has a correct label (see Predicate $CorrectF(v)$), then every node $v \in F$ is able to determine if $F$ spans all the nodes of the network or not. This knowledge is given by the label of the nodes, more precisely using the decoder given in Subsection~\ref{subsec:decoder}. Indeed, given a non-tree edge $e=\{u,v\}$, if the nodes $u$ and $v$ have no common ancestor then $u$ and $v$ are in two distinct fragments. In this case, the merging phase can be executed at $u$ and $v$. A merging phase is composed of several \emph{merging steps} in which at least two fragments are merged. Each merging step is performed following four steps: 
\begin{enumerate}
\item The root of each fragment $F$ identifies the minimum-weight outgoing edge $e=(u,v)$ of its fragment (see Rule $\RMin$).
\item After the computation of $e$ each node $x$ on the path between the root of $F$ and $v \in F$ computes in variable $newp_x$ its future parent (see rule $\RF$). The nodes in the sub-tree rooted at every node $x$ executes also Rule $\RF$.
\item When the two merging fragments have finished the two first steps, then each node of these two fragments can compute their future distance (see Rule $\RDist$).
\item Finally, every node $v$ belonging to these two fragments copies the content of its variables $newp_v$ (resp. $newd_v$) into variable $\parent_v$ (resp. $d_v$).
\end{enumerate}

Let us proceed with a more detailed description of the these steps. We process the computation of the minimum-weight outgoing edge of each fragment in a bottom-up manner (see Rule $\RMin$). Each node $v$ can identify its adjacent outgoing edges $e=\{u,v\}$ by computing locally that $e$ has no nearest common ancestor using the labels of $u$ and $v$. This is done via the decoder given in subsection~\ref{subsec:decoder} and Macro $\OEd(v)$ at $v$. Each node $v$ computes the minimum-weight outgoing edge of its sub-tree (given by Macro $\Fusion(u)$) by selecting the edge of minimum-weight among its adjacent outgoing edges (given by Macro $\FA(v)$) and the one given by its children (given by Macro $\FC(v)$). The weight and the identifier of the extremities of the minimum-weight outgoing edge are stored in variable $\m_v$ at $v$. All these information will be used for the merging step. Figure~\ref{fig:fusion} depicts the selection of the minimum outgoing edge for two fragments.

\begin{figure}[!ht]
\begin{center}
 \includegraphics[scale=0.55]{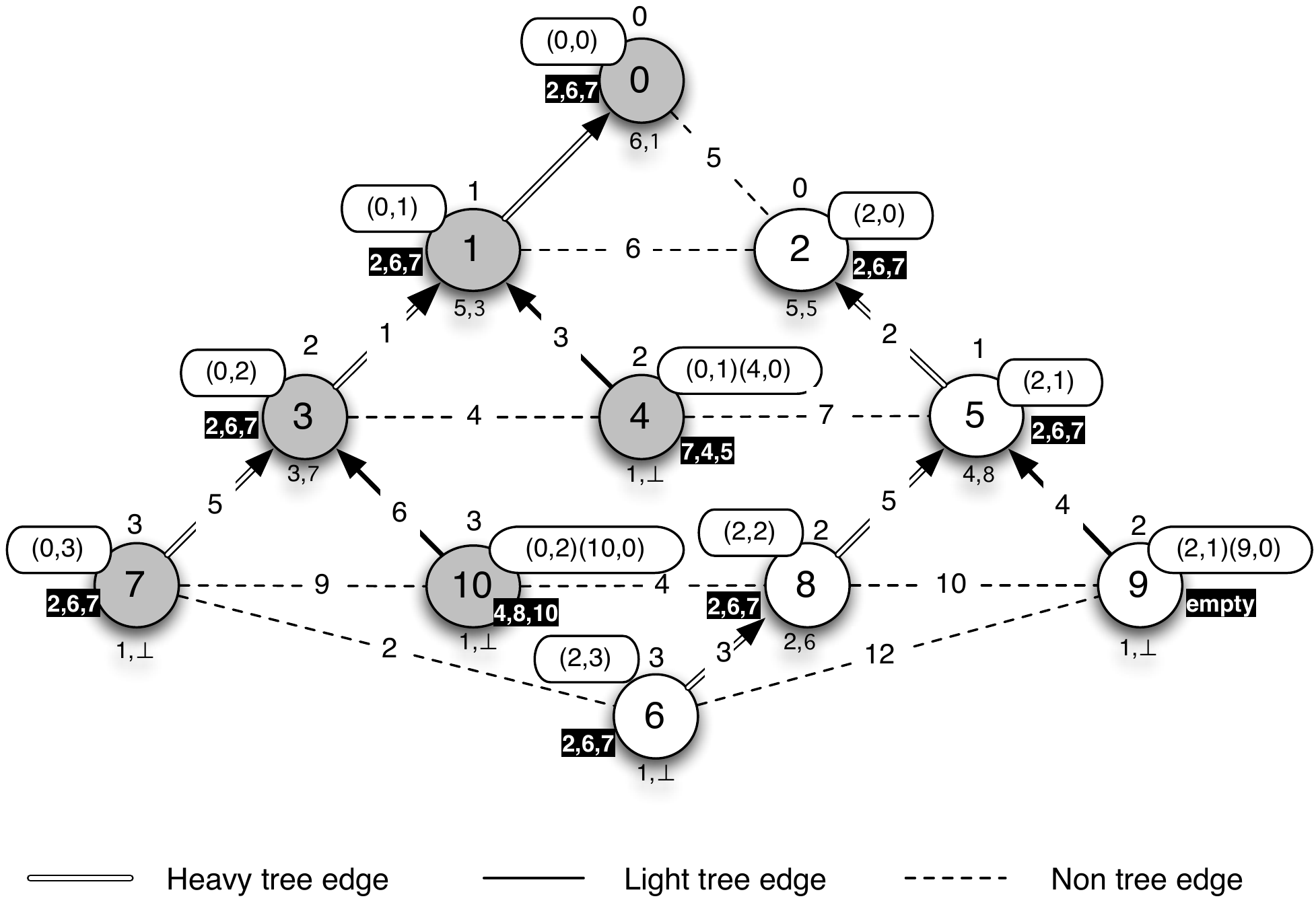}
\caption{Minimum weight outgoing edge computation based on Nearest Common Ancestor Labeling scheme for a forest:  The white bubble at each node $v$ corresponds to the label of the node. The black bubble at each node represent the selection of minimum outgoing edge. The information under the node corresponds to the variable $\size$ and the information on top of the node represent the distance of the node from the root.}
\label{fig:fusion}
\end{center}
\end{figure}
\bigskip




\begin{figure}[ht]
\fbox{
\begin{minipage}{16cm}
\begin{tabular}{lll}
$\Child(v)$ & $=$ & $\{u \in \N_v : \parent_u=\id_v\}$\\
\vspace*{0.2cm}$\OEd(v)$ & $=$ & $\min\{(u,v): u \in N_v \setminus (\Child(v) \cup \{\parent_v\}) \wedge \Lca(\lab_u,\lab_v)=\emptyset\}$\\
\vspace*{0.2cm}$NCand(val)$ & $=$ & $\left \{ \begin{array}{ll} \min\{(u,v) \in \OEd(v) : w(u,v)=val\} & \mbox{If } val=\FC(v)\\ \min\{\m_u[1]: u \in \Child(v) \wedge \m_u[0]=val\} & \mbox{Otherwise} \end{array} \right .$\\
$\FC(v)$ & $=$ & $\min\{ \m_u[0]:u\in \Child(v)\}$\\
$\FA(v)$ & $=$ & $\min\{w\{v,u\}: (u,v) \in \OEd(v)\}$\\
\vspace*{0.2cm}$\Fusion(v)$ & $=$ & $\min\{\FC(v),\FA(v)\}$\\
\vspace*{0.2cm}$NewParent(v)$ & $=$ & $\left \{ \begin{array}{ll} \min\{u \in N_v: (u,v)=\m_v[1]\} & \mbox{If } \m_v[0]=\FA(v) \wedge \m_v[1] \in \OEd(v)\\ \min\{u \in \Child(v): \m_u=\m_v\} & \mbox{Otherwise} \end{array} \right .$\\
\vspace*{0.2cm}$NewDist(v)$ & $=$ & $\left \{ \begin{array}{ll} 0 & \mbox{If } newp_{newp_v}=v \wedge \id_{newp_v}>\id_v\\ 1 & \mbox{If } newp_{newp_v}=v \wedge \id_{newp_v}<\id_v\\ newd_{newp_v}+1 & \mbox{Otherwise} \end{array} \right .$\\
$NewChild(v)$ & $=$ & $\{u: (u \in (C(v) \cup \{\parent_v\}) \wedge newp_u=v)$\\
& & \hspace*{0.8cm}$\vee (u \in N_v \wedge newp_u=v \wedge newp_v=u \wedge \id_{newp_v}>\id_v)\}$\\
\end{tabular}
\begin{tabular}{lll}
$CorrectF(v)$ & $\equiv$ & $\Distance(v) \wedge \SizeC(v) \wedge \Label(v)$\\
$Reorientation(v)$ & $\equiv$ & $\parent_v=\emptyset \vee (\m_{\parent_v}=\m_v \wedge newp_{\parent_v}=v)$\\
$ChangeNewP(v)$ & $\equiv$ & $(Reorientation(v) \wedge newp_v \neq NewParent(v)) \vee (\m_{\parent_v} \neq \m_v \wedge newp_v \neq \parent_v)$\\
$ChangeNewD(v)$ & $\equiv$ & $(newp_{newp_v}=v \wedge newd_v>1) \vee (\parent_{newp_v}=v \wedge newd_v \neq newd_{newp_v}+1)$\\
& & $\vee (newd_{\parent_v} \neq d_{\parent_v} \wedge newd_v \neq newd_{\parent_v}+1)$\\
$CopyVar(v)$ & $\equiv$ & $(\forall u \in NewChild(v), d_u=newd_u \wedge \parent_u=newp_u)$\\
& & $\wedge (\parent_v \neq newp_v \vee d_v \neq newd_v)$\\
\end{tabular}\\
\end{minipage}
}
\caption{Macros and predicates used by Algorithm \MSTA\/ for the merging.}
\label{fig:predicatesb}
\end{figure}

When the computation of the minimum-weight outgoing edge $e=(u,v)$ is finished at the root $r$ of a fragment $F$ (i.e., $\m_r=\Fusion(r)$), then $r$ can start the computation of the future parent pointers in $F$ (Predicate $ChangeNewP(r)$ is satisfied), done in a top-down manner (see rule $\RF$). Let $u$ be the extremity of $e$ of minimum identity between $u$ and $v$. If $e$ is selected as the minimum-weight outgoing edge of two fragments $F$ and $F'$, then $u$ will become the new root of Fragment $F''$ resulting from the merging between $F$ and $F'$. Otherwise, $e$ is the minimum-weight outgoing edge selected only by a single fragment, w.l.o.g. let $F$. In this case, $F$ will wait for that $e$ is selected as the minimum-weight outgoing edge of $F'$. In the two cases, every node $v$ of a fragment in a merging step computes its future parent pointer in variable $\m_v$. Each node on the path from the root of the fragment leading to the minimum-weight outgoing edge selects its child on this path as its future parent, while the other nodes select their actual parent.

When $e$ is selected as the minimum-weight outgoing edge by $F$ and $F'$ and the computation of the future parent is done (i.e., $\neg ChangeNewP(v)$ is satisfied), then the future distance is computed in variable $newd_v$ by each node $v$ in $F \cup F'$ (Predicate $ChangeNewD(v)$ is satisfied), in a top-down manner following the parent relation given by variable $newp_v$ (see Rule $\RDist$). Note that the extremity of $e$ with the minimum identifier becomes the root of the new fragment with a zero distance. Finally, when the future parent and distance are computed by every node $v$ in $F \cup F'$ then $v$ can execute Rule $\RFEnd$ (see Predicate $CopyVar(v)$) to copy the content of variable $newp_v$ (resp. $newd_v$) into variable $\parent_v$ (resp. $d_v$). Note that this is done in a bottom-up fashion following the parent relation given by variable $newp_v$ in order to not destabilize Fragment $F$ or $F'$.

\hrulefill
\begin{description}
\item $\RMin$: [ \textbf{Minimum computation} ] \\
\hspace*{0,2cm}\textbf{If} $CorrectF(v) \wedge \big(\m_v[0] \neq \Fusion(v) \neq \emptyset \big)  $ \textbf{Then}\\
\hspace*{0,8cm} $\m_v:=(\Fusion(v), NCand(\Fusion(v)));$\\
\item $\RF$: [ \textbf{Merging} ] \\
\textbf{If} $CorrectF(v) \wedge \big(\m_v[0] = \Fusion(v) \neq \emptyset \big) \wedge ChangeNewP(v)  $ \textbf{Then}\\
\hspace*{0,5cm} $newd_v:=\infty;$\\
\hspace*{0,5cm} \textbf{If} $Reorientation(v)$ \textbf{Then} $newp_v:=NewParent(v);$\\
\hspace*{0,5cm} \textbf{Else} $newp_v:=\parent_v$\\
\item $\RDist$: [ \textbf{New distance} ] \\
\textbf{If} $CorrectF(v) \wedge \neg ChangeNewP(v) \wedge ChangeNewD(v)$ \textbf{Then}\\
\hspace*{0,5cm} $newd_v:=NewDist(v);$\\
\item $\RFEnd$: [ \textbf{End of merging} ] \\
\textbf{If} $\Distance(v) \wedge \neg ChangeNewP(v) \wedge \neg ChangeNewD(v) \wedge CopyVar(v)$ \textbf{Then}\\
\hspace*{0,5cm} $\parent_v:=newp_v; d_v:=newd_v; \m_v:=\emptyset;$\\
\hspace*{0,5cm} \textbf{If} $newd_v=0$ \textbf{Then} $\parent_v:=\emptyset; newp_v:=\emptyset;$
\end{description}
\hrulefill

\subsection{Recovering phase}

This subsection is dedicated to the description of the recovering phase. Recall that, since the system can start from an arbitrary configuration $\gamma$, edges which do not belong to any {\tt MST} can be part a fragment in $\gamma$. Given a fragment $F$, the addition of an edge $e=(u,v)$ which do not belong to $F$ creates a unique cycle, called \emph{fundamental cycle} related to $e$ and denoted $C_e$ (i.e., $C_e=\path(u,\Lca(\lab_u,\lab_v))\cup \path(\Lca(\lab_u,\lab_v),v)\cup \{e\}$). Thus, the "red rule" may not be satisfied for every constructed fragment, i.e., for some fundamental cycle defined by an internal edge of a fragment the maximum edge weight belong to the fragment. To identify these edges, we verify that for each internal edge $e$ there is no edge in the fundamental cycle $C_e$ with a higher weight than $w(e)$. To this end, in a fragment $F$ the label of the nodes are used to identify the edges $e=\{u,v\}$ which do not belong to $F$ such that $u$ and $v$ have a common ancestor.

\begin{figure}[!tbh]
\centering
\subfigure[]{
\includegraphics[width=.25\textwidth]{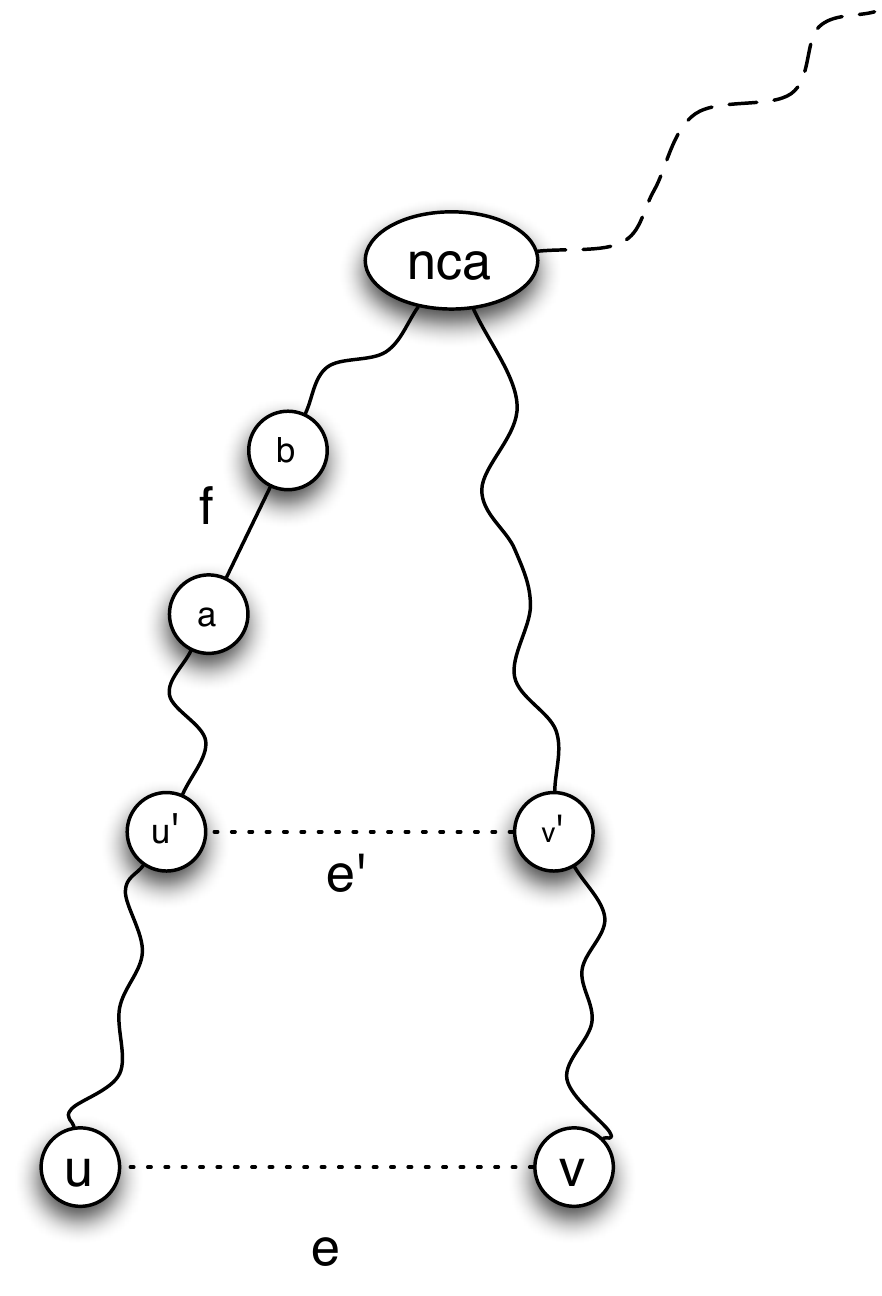}
\label{fig:IdemNCA}
}
\subfigure[]{
\includegraphics[width=.24\textwidth]{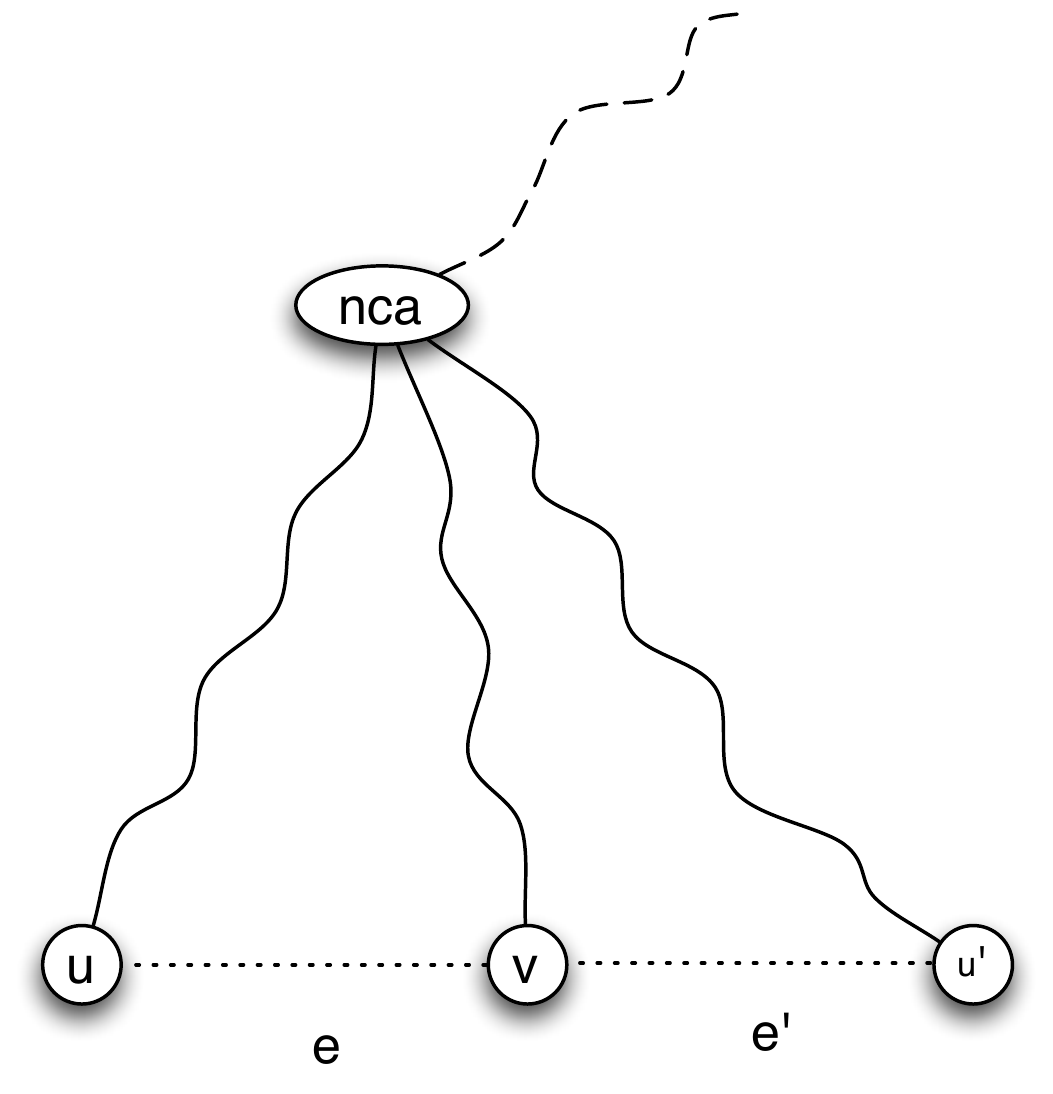}
\label{fig:IdemNCA2}
}
\subfigure[]{
\includegraphics[width=.435\textwidth]{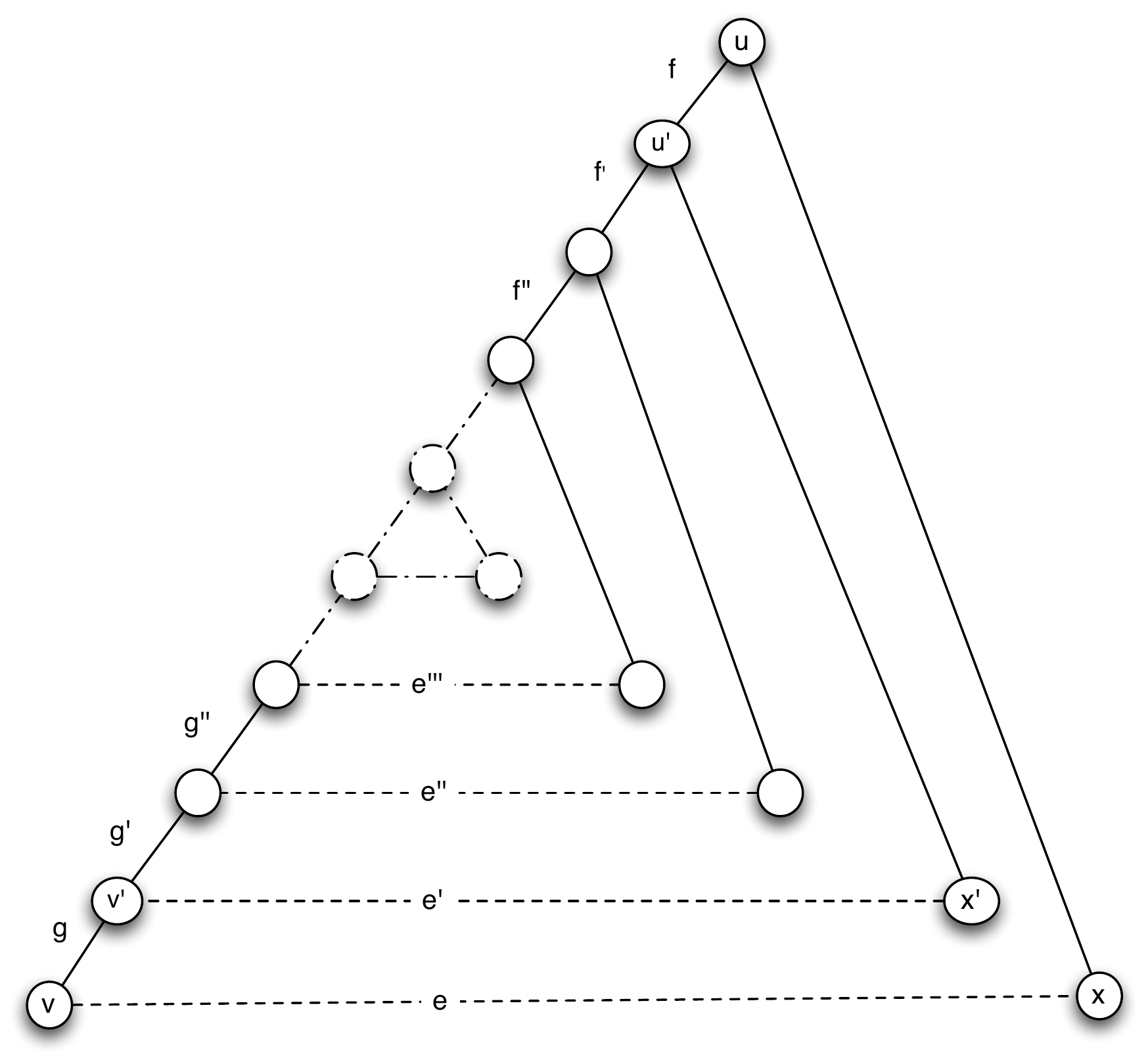}
\label{fig:UpEdges}
}
\caption{Use of internal edges for fundamental cycles verification}
\end{figure}

Let us consider a fragment $F$ and an edge $f=\{x,y\}$ belonging to $F$ such that $f\in C_e\setminus\{e\}$. If $w(f)>w(e)$ then $e$ must become an edge of $F$. Consequently, we need to verify all the edge weights of $C_e\setminus\{e\}$. To achieve this task, the weight of $e=(u,v)$ is sent up in $F$ along the two paths $\path(u,\Lca(\lab_u,\lab_v))$ and $\path(\Lca(\lab_u,\lab_v),v)$. Clearly, to maintain low space complexity, the nodes cannot store the information about all internal edges. Consequently, we decide that each node stores only the information of a single internal edge at a time. Specifically, we need to organize the circulation of the internal edges. A natural question to ask at this point is whether the information of all non-tree edges are needed. To answer to this question, we first make some observations.

First, suppose the following case (see Figure~\ref{fig:IdemNCA}): let $e=\{u,v\}$ and $e'=\{u',v'\}$ be internal edges such that \mbox{$\Lca(\lab_u,\lab_v)=\Lca(\lab_{u'},\lab_{v'})$}, and $u'$ and $v'$ are closer to $\Lca(\lab_u,\lab_v)$ than $u$ and $v$. On $\path(u',\Lca(\lab_u,\lab_v))$ and $\path(v',\Lca(\lab_u,\lab_v))$ only the internal edge with the smallest weight is needed. To justify this assertion, let us consider without loss of generality that $w(e)<w(e')$ and $f=\{a,b\}$ is a tree edge such that $w(f)>w(e')$. Moreover, suppose that all edges in a $\path(u,u')$ and $\path(v,v')$ have a weight smaller than $w(e)$. Consequently, $f$ is not part of the {\tt MST}, and if we delete $f$, the minimum outgoing edge of the fragment composed by the $\path(a,u)$ is edge $e$. Consider now, the case when several adjacent edges of node $v$ have the same common ancestor (see Figure~\ref{fig:IdemNCA2}). In this case only the internal edge with the smallest weight is relevant on the $\path(v,\Lca(\lab_u,\lab_v))$ to avoid the maximum weight of the fundamental cycles. The last case considered is the following (see Figure~\ref{fig:UpEdges}). Consider a path between two nodes $u$ and $v$, and $u',v'\in\path(u,v)$ such that $f=\{u,u'\}$ and $g=\{v,v'\}$. Let $e=\{v,x\}$ be an edge such that $\Lca(\lab_v,\lab_x)=u$ and $e'=\{v',x'\}$ an edge such that $\Lca(\lab_{v'},\lab_{x'})=u'$. If $w(e')<w(e)$, the weight of $e'$ is needed to verify if the weights of the edges on $\path(v',u')$ have a higher weight than $e'$. However, the weight of $e$ is needed to verify the weight of edge $f$. Consequently, we need to collect all the outgoing edges from the leaves to the root, from the farthest to the nearest of the root.


\begin{figure}[!tbh]
\fbox{
\begin{minipage}{16cm}
\begin{tabular}{lll}
$IECA(v,ca)$ & $=$ & $\min\{w(u,v): u \in N_v \setminus (\Child(v) \cup \{\parent_v\}) \wedge \Lca(\lab_u,\lab_v)=ca \wedge \Lca(\lab_u,\lab_v) \succeq \lab_v\}\footnote{ Operator $\succ$ is the lexicographical order used for the node labels, we consider $\bot$ as the smallest element.}$\\
$\IE_l(v)$ & $=$ & $\{(w(u,v),\lab_u,\lab_v): u \in N_v \setminus (\Child(v) \cup \{\parent_v\}) \wedge \Lca(\lab_u,\lab_v) \neq \emptyset$\\
& & $\wedge w(u,v)=IECA(v,\Lca(\lab_u,\lab_v)\}$\\
$\IE_c(v)$ & $=$ & $\{\mIns_u: u \in \Child(v) \wedge \Lca(\mIns_u[1],\mIns_u[2]) \succeq \lab_v\}$\\
\vspace*{0.2cm}$minIE(v,ca)$ & $=$ & $\min\{e: e \in (\IE_l(v) \cup \IE_c(v)) \wedge \Lca(e[1],e[2]) \succ ca\}$\\
$\IE(v)$ & $=$ & $\left\{ \begin{array}{ll} minIE(v, \Lca(\mIns_v[1],\mIns_v[2])) & \mbox{If } minIE(v, \Lca(\mIns_v[1],\mIns_v[2])) \neq \emptyset\\ minIE(v, \bot) & \mbox{Otherwise} \end{array} \right.$
\end{tabular}
\begin{tabular}{lll}
$EndForward(v)$ & $\equiv$ & $(\parent_v=\emptyset \vee \Lca(\mIns_v[1],\mIns_v[2])=\lab_v) \wedge \mIns_v=\IE_l(v)$\\
$Forwarded(v)$ & $\equiv$ & $(\Lca(\mIns_v[1],\mIns_v[2])=\lab_v \vee \mIns_{\parent_v}=\mIns_v) \wedge \mIns_v \not \in \IE_c(v)$\\
$BetterEdgeP(v)$ & $\equiv$ & $\parent_v \neq \emptyset \wedge \Lca(\mIns_v[1],\mIns_v[2]) \neq \lab_v$\\
& & $\wedge \Lca(\mIns_{\parent_v}[1],\mIns_{\parent_v}[2])=\Lca(\mIns_v[1],\mIns_v[2]) \wedge \mIns_{\parent_v}<\mIns_v$\\
$SelectEdge(v)$ & $\equiv$ & $EndForward(v) \vee Forwarded(v) \vee BetterEdgeP(v)$\\
$Recover(v)$ & $\equiv$ & $\parent_v=newp_v \wedge d_v=newd_v \wedge newd_v=newd_{\parent_v}+1 \wedge \Fusion(v)=\emptyset$\\
& & $\wedge SelectEdge(v)$
\end{tabular}
\end{minipage}
}
\caption{Macros used by Algorithm \MSTA\/ for the correction of the MST.}
\label{fig:predicatesIn}
\end{figure}

The rule for collecting the relevant internal edges is based on the above observations (see rule $\RC$). The internal edges are sent up in the fragment from the leaves to the root using variable $\mIns_v$ at every node $v$. The internal edges are collected locally by beginning from the edge with the farthest nearest common ancestor to the edge with the nearest common ancestor, i.e., following the lexicographical order on the nearest common ancestor labels and beginning by the smallest one. In case there exist several edges with the same nearest common ancestor, only the edge with the smallest weight is kept. The list of the ordered internal edges at node $v$ is given by Macro $\IE(v)$. This list is computed by different predicates (see Macros in Figure~\ref{fig:predicatesIn}). Each node $v$ compares the weight stored in variable $\mIns_v[0]$ with the weight of the edge leading to its parent. If $\mIns_v[0]<w(v,\parent_v)$ then $v$ knows that the internal edge indicated by $\mIns_v$ must belong to the {\tt MST}. Consequently $v$ deletes the edge $(v,\parent_v)$ from the fragment (only if $v$ is not the nearest common ancestor of the internal edge given by $\mIns_v$), and $v$ becomes the root of the new fragment (see Rule $\RC$). A node $v$ can select a new internal edge by executing Rule $\RC$ in the following case (i.e., Predicate $Recover(v)$ is satisfied): (i) the internal edge of $v$ is propagated up by its parent and $v$ has no more child propagating the same internal edge (see Predicate $Forwarded(v)$), (ii) $v$ is the nearest common ancestor of the adjacent internal edge actually selected (see Predicate $EndForward(v)$), or (iii) $v$ is neither the root of the fragment nor the nearest common ancestor of the selected internal edge $e'$ and its parent propagates an internal $e''$ related with the same common ancestor but $w(e'')<w(e')$ (see Predicate $BetterEdgeP(v)$). This allows to obtain a piplined propagation of the internal edges. Figure~\ref{fig:Recovering} illustrates the bottom-up spread of the internal edges.


\bigskip
\begin{minipage}{17cm}
\hrulefill
\begin{description}
\item[$\RC$: [\ Recovering]] 
\item \textbf{If} $CorrectF(v) \wedge Recover(v)$ \textbf{Then}
\item  \hspace*{0.5cm} $\mIns_v:=\IE(v);$
\item  \hspace*{0.5cm}  \textbf{If} $\Lca(\mIns_v[1],\mIns_v[2]) \neq \lab_v \wedge w(v,\parent_v)>\mIns_v[0]$ \textbf{Then} $\parent_v:=\emptyset; d_v:=0; \lab_v:=(\id_v,0);$
\end{description}
\hrulefill
\end{minipage}

\begin{figure}[!tbh]
\begin{center}
\includegraphics[scale=0.55]{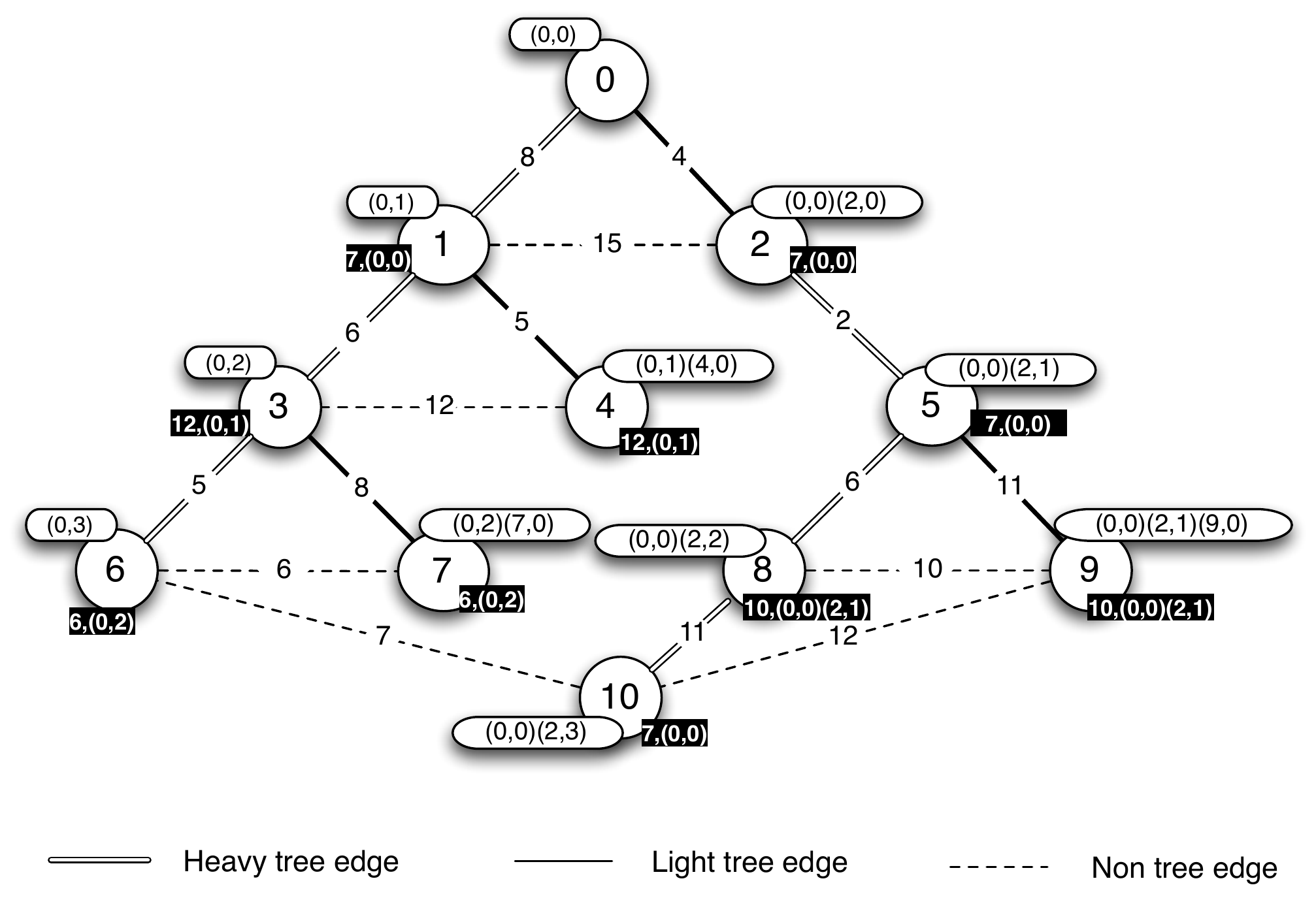}
\caption{The white bubble at each node $v$ corresponds to the label of the node. The black bubble at each node represents the internal edges.}
\label{fig:Recovering}
\end{center}
\end{figure}
\bigskip

\subsection{Correctness and complexity}

This subsection is dedicated to the correctness of the self-stabilizing Minimum Spanning Tree construction. We can define a Minimum Spanning Tree as in Definition~\ref{def:mst}.

\begin{definition}[MST]
\label{def:mst}
Let $G=(V,E,w)$ be a network with $V$ the set of nodes, $E$ the set of undirected links and the function $w: E \rightarrow \mathbb{N}$. A graph $T=(V_T,E_T)$ of $G$ is called a \emph{Minimum Spanning Tree} if the following conditions are satisfied:
\begin{enumerate}
\item $V_T=V$ and $E_T \subseteq E$, and
\item $T$ is a connected graph (i.e., there exists a path in $T$ between any pair of nodes $x,y \in V_T$) and $|E_T|=|V|-1$, and
\item There exists no spanning tree $T'$ of $G$ whose the weight $w(T')$ is lower than $w(T)$.
\end{enumerate}
\end{definition}

We give a formal specification to the problem of constructing a Minimum Spanning Tree, stated in Specification~\ref{spec:mst}.

\begin{specification}[MST Construction]
\label{spec:mst}
Let $\Gamma$ be the set of all possible configurations of the system. An algorithm $\mathcal{A_{MST}}$ solving the problem of constructing a stabilizing MST tree satisfies the following conditions:
\begin{itemize}
\item[\emph{[TC1]}] Starting from any configuration in $\Gamma$, Algorithm $\mathcal{A_{MST}}$ reaches in finite time a set of configurations $\mathcal{L} \subseteq \Gamma$ which satisfies Definition~\ref{def:mst}, and
\item[\emph{[TC2]}] From every configuration $\gamma \in \mathcal{L}$, Algorithm $\mathcal{A_{MST}}$ can only reach a configuration in $\mathcal{L}$.
\end{itemize}
\end{specification}

Let $\Gamma$ be the set of all possible configurations of the system. A fragment ${\tt F}$ rooted at node $r_{\tt F}$ is a subtree such that for every node $v \in {\tt F}$ there is a path to $r_{\tt F}$ and Predicate $\Distance(v)$ is true. In the following theorem we start by showing that until a legitimate configuration is reached there is no deadlock in the system.

\begin{theorem}
\label{thm:no_deadlock}
Let the set of configurations $\mathcal{B} \subseteq \Gamma$ such that every configuration $\gamma \in \mathcal{B}$ satisfies Definition~\ref{def:mst}. $\forall \gamma \in (\Gamma -\mathcal{B}, \exists v \in V$ such that $v$ is enabled in $\gamma$.
\end{theorem}

\begin{proof}
Assume by the contradiction, that $\exists \gamma \in (\Gamma - \mathcal{B})$ such that $\forall v \in V$ no rule is enabled at $v$ in $\gamma$. Since $\gamma \not \in \mathcal{B}$, there is either a cycle, several fragments, or a single fragment which is not a MST in $\gamma$. If there is a cycle or incorrect distances in $\gamma$ then there exists a node $v$ such that $d_v \neq d_{\parent_v}+1$. This implies that Predicate $\Distance(v)$ is not satisfied and Rule $\RCo$ is enabled at $v$, a contradiction. Otherwise, $\forall v \in V$ Predicate $Distance(v)$ is satisfied. If there exists a node $v$ in $\gamma$ with an incorrect label, then either Predicate $\SizeC(v)$ or $\Label(v)$ is satisfied and Rule $\RSC$ or $\RLC$ is enabled at $v$ (see proofs of Section~\ref{subsec:CorLabel} for more details), a contradiction. Otherwise, Predicate $CorrectF(v)$ is satisfied $\forall v$ in $\gamma$. If there are several fragments in $\gamma$ then there is at least one node $v \in V$ such that Macro $C\Fusion(v)\neq \emptyset$. If there is a node $v$ in $\gamma$ which has not computed the correct outgoing edge of its subtree (i.e., $\m_v[0] \neq \Fusion(v)$), then Rule $\RMin$ is enabled at $v$, a contradiction. Otherwise, in each fragment $F$ in $\gamma$ we have $\forall v \in F, \m_v[0]=\Fusion(v)$. Consider first a node $v$ in a fragment $F$ in $\gamma$ which is on the path between the root of $F$ and the minimum outgoing edge of $F$ (i.e., $\m_{\parent_v}=\m_v$). If there exists such a node $v$ with $newp_v \neq NewParent(v)$ and Predicate $Reorientation(v)$ is satisfied, then Predicate $ChangeNewP(v)$ Rule $\RF$ is enabled at $v$, a contradiction. Otherwise, consider the other node $v$ in $F$ which are not on the path between the root and the minimum outgoing edge of $F$ (i.e., $\m_{\parent_v} \neq \m_v$). If there exist such a node $v$ such that $newp_v \neq \parent_v$ then Predicate $ChangeNewP(v)$ is satisfied and Rule $\RF$ is enabled at $v$, a contradiction. Otherwise in each fragment $F$ in $\gamma$, we have $\forall v \in V, (newp_v=NewParent(v) \vee newp_v=\parent_v) \Rightarrow \neg ChangeNewP(v)$. Either for the future root $v$ (i.e., $newp_{newp_v}=v$) of a fragment $F$ in $\gamma$ we have $newd_v>1$ then Predicate $ChangeNewD(v)$ is satisfied and Rule $\RDist$ is enabled at $v$, a contradiction. Or for the other node $v$ in $F$ we have $newd_v \neq newd_{\parent_v}+1$ then Predicate $ChangeNewD(v)$ is satisfied and Rule $\RDist$ is enabled at $v$, a contradiction. Otherwise, we have $\forall v \in V, \neg ChangeNewP(v) \wedge ChangeNewD(v)$ in $\gamma$. If in a fragment $F$ in $\gamma$ there is a node $v$ such that every of its future children $u$ after the merging (given by Macro $NewChild(v)$) in the fragment satisfies $\neg CopyVar(u)$ and $\parent_v \neq newp_v \vee d_v \neq newd_v$, then Predicate $CopyVar(v)$ is satisfied and Rule $\RFEnd$ is enabled at $v$, a contradiction. Finally, otherwise there is only a single fragment $F$ in $\gamma$ and we have $\forall v \in F, CorrectF(v)$. Moreover, for every node $v \in V$ Predicate $Recover(v)$ is satisfied, since $(\neg ChangeNewP(v) \wedge \neg ChangeNewD(v) \wedge \neg CopyVar(v) \wedge CorrectF(v)) \Rightarrow Recover(v)$. Therefore, Rule $\RC$ is enabled at every $v$ in $\gamma$. By contradiction the fragment $F$ in $\gamma$ is not a MST, so there exists a node $v$ in $\gamma$ such that $\parent_v \neq \emptyset$ and $v$ is adjacent of an internal edge with a weight lower than $w(v,\parent_v)$ (i.e.,$w(v,\parent_v)>\mIns_v[0]$). Thus, $v$ becomes the root of a new fragment when $v$ executes Rule $\RC$, a contradiction.
\end{proof}

We denote by $\Gamma_{\tt CF}$ the set of configurations in $\Gamma$ such that there are no cycles in the subgraph induced by parent link relations (i.e., for every $\gamma \in \Gamma_{\tt CF}$ we have $\forall v \in V, \Distance(v)$). 

\begin{lemma}
\label{lem:time_cf}
Starting from any arbitrary configuration, the system reaches in $O(n)$ rounds a configuration in $\Gamma_{\tt CF}$.
\label{lem:cycle free}
\end{lemma}

This lemma can be proved using the same arguments given in~\cite{DolevIM97}.\\

\subsubsection{Correctness and complexity of the merging phase}

We now define some notations and predicates which will be used in the following proofs. Given a configuration $\gamma \in \Gamma$, we note the set of all fragments in $\gamma$ by $\mathcal{F}(\gamma)$. 
Moreover, we define below several sets of fragments with different properties and the notion of \emph{attractor}, introduced by Gouda and Multari~\cite{GoudaM91}, will be used to show that during the convergence of Algorithm $\MSTA$ each fragment gains additional properties. We define five sets of fragments in a configuration $\gamma \in \Gamma$:

\begin{itemize}
\item Let $\mathcal{F}_1(\gamma)=\{F \in \mathcal{F}(\gamma): (\forall v \in F: CorrectF(v))\}$ be the set of fragments in $\gamma$ in which all the nodes are correctly labeled.
\item Let $\mathcal{F}_2(\gamma)=\{F \in \mathcal{F}_1(\gamma): (\forall v \in F: \m_v[0]=\Fusion(v) \neq \emptyset)\}$ be the set of fragments correctly labeled in $\gamma$ in which every node has computed its minimum-weight outgoing edge of its subtree for the merging phase.
\item Let $\mathcal{F}_3(\gamma)=\{F \in \mathcal{F}_2(\gamma): (\forall v \in F: \neg ChangeNewP(v))\}$ be the set of fragments correctly labeled in $\gamma$ in which every node has computed its future parent used when the merging phase is done.
\item Let $\mathcal{F}_4(\gamma)=\{F \in \mathcal{F}_3(\gamma): (\forall v \in F: \neg ChangeNewD(v))\}$ be the set of fragments correctly labeled in $\gamma$ in which every node has computed its future distance used when the merging phase is done.
\item Let $\mathcal{F}_5(\gamma)=\{F \in \mathcal{F}_4(\gamma): (\forall v \in F: \neg CopyVar(v))\}$ be the set of fragments in $\gamma$ for which the merging phase is done.
\end{itemize}

We obtain the following lemma by applying Rules $\RSC$ and $\RLC$ according to Lemmas~\ref{lem:size_convergence} to~\ref{lem:label_closure} and Theorem~\ref{thm:self-stab_lab}.

\begin{lemma}
\label{lem:labelMST}
Starting from any configuration $\gamma\in \Gamma_{\tt CF}$, after $O(n)$ rounds the system reaches a configuration $\gamma'$ such that for each fragment $F \in \mathcal{F}_{\gamma}$ we have $F \in F_1(\gamma')$.
\end{lemma}

\begin{lemma}
\label{lem:Min}
Let any fragment $F \in \mathcal{F}_1(\gamma)$ in a configuration $\gamma \in \Gamma_{\tt CF}$. In $O(h_{\mathcal{F}})$ rounds, we have $F \in \mathcal{F}_2(\gamma')$, with $\gamma' \in \Gamma_{\tt CF}$ and $h_F$ the height of $F$.
\label{lem:min_convergence}
\end{lemma}

\begin{proof}
In the following we define the potential function ${\cal M}$. First, let $w_m$:  $V\rightarrow \mathbb{N}$ be the function defined by: 
$$w_m(v)=|\m_v[0]-\min\Big (\min\{\m_u[0]:u\in\Child(v)\}, \min\{w\{v,u\}: (u,v) \in \OEd(v)\}\Big )|.$$
Note that, we have $w_m \geq 0$. Variable $\m_v$ has a correct value at node $v$ if and only if $w_m=0$. Let a fragment $F \in \mathcal{F}_1(\gamma)$ in a configuration $\gamma \in \Gamma_{\tt CF}$, and ${\cal M}: \Gamma\rightarrow\mathbb{N}$ be the function defined by:
$${\cal M}(\gamma)=\mathlarger{\sum}_{\mathsmaller {d=0}}^{\mathsmaller{h_{\tt f}}} m_d(\gamma) (n+1)^d$$
where $m_d(\gamma)$ is the number of nodes $v$ at height $d$ in $F$ with  $w_m(v)\neq0$. We denote $n_F$ the number of nodes in fragment $F$. Note that $0\leq m_d(\gamma) \leq n_F$, and $0\leq {\cal M}(\gamma)\leq (n_F+1)^{h_F+1}$. Moreover, the variable $\m$ has a correct value at every node in $F$ if and only if  ${\cal M}(\gamma)=0$. We note $\gamma(t)$ the configuration of the system after round $t$. Let $d_0$ be the largest index such that $m_{d_0}(\gamma(t))\neq 0$. Since we use a weakly fair scheduler, all the nodes of Fragment $F$ are scheduled during the execution of round $t+1$. Every node $v$ at height $d>d_0$ does not change the value of its variable $\m$ (see Rule $\RMin$), and therefore $w_m(v)$ remains equal to zero, so $m_d(\gamma(t+1))$ is equal to zero as well. The nodes $v$ at height $d_0$ change their variable $\m_v$ according to the variable $\m_u$ of their children $u \in F$ (see Rule $\RMin$). Let $v$ be a node at height $d_0$.  The children of $v \in F$ (if any) are at height $d>d_0$. Thus, their variable $\m$ has not changed, and therefore $w_m(v)$ becomes zero after round $t+1$. As a consequence, $m_{d_0}(\gamma(t+1))=0$. Therefore, we get 
$${\cal M}(\gamma(t+1))<{\cal M}(\gamma(t))$$
and  thus the system will eventually reach a configuration where all the variables $\m$ contains the minimum outgoing edge of the sub-fragment rooted at $v \in F$ (see Predicate $\Fusion(v)$).

To measure the number of rounds it takes to converge, observe that $d$ decreases by at least one at each round. Since $d \leq h_F$, we get that starting from any  configuration $\gamma \in \Gamma_{\tt CF}$ with $F \in \mathcal{F}_1(\gamma)$ the system reaches a configuration where for every node $v$ in Fragment $F$ the variable $\m_v$ is correct after $O(h_F)$ rounds.
\end{proof}

\begin{lemma}
\label{lem:merging_two_frag}
In every configuration $\gamma \in \Gamma_{\tt CF}$, there are at least two fragments $F_1$ and $F_2$, $F_1,F_2 \in \mathcal{F}_2(\gamma)$ which select the same minimum-weight outgoing edge for merging.
\end{lemma}

\begin{proof}
Assume, by the contradiction, that there exists a configuration $\gamma \in \Gamma_{\tt CF}$ with less than two fragments in $\mathcal{F}_2(\gamma)$ which select the same minimum-weight outgoing edge. This implies in $\gamma$ that either at least one fragment $F \in \mathcal{F}(\gamma)$ which has not computed its minimum-weight outgoing edge, or every fragment $F \in \mathcal{F}_2(\gamma)$ has selected a different minimum-weight outgoing edge. In the former case, there is a contradiction since according to Lemma~\ref{lem:Min} in $O(h_F)$ additional rounds the system reaches a configuration $\gamma'$ in which at least two fragments $F_1$ and $F_2$, $F_1,F_2 \in \mathcal{F}_2(\gamma)$ which select the same minimum-weight outgoing edge for merging. Otherwise, let $|\mathcal{F}_2(\gamma)|$ denotes the number of fragments in the set $\mathcal{F}_2(\gamma)$ in $\gamma$. In the latter case, exactly $|\mathcal{F}_2(\gamma)|$ minimum-weight outgoing edges have been selected in $\gamma$. However, we can observe that we can define a total order on the outgoing edges in each configuration in $\Gamma_{\tt CF}$ based on the tuple defined by the edges weight and the identifiers of the extremities of the edges. By using this total order, $n$ fragments could select at most the $n-1$ minimum outgoing edges. Thus, since Algorithm $\MSTA$ uses these method to select the minimum-weight outgoing edge of each fragment (see Macros $\Fusion(v)$ and $NCand(\Fusion(v))$) then at most $|\mathcal{F}_2(\gamma)|-1$ different outgoing edges are selected in $\gamma$. So, there are at least two fragments $F_1,F_2 \in \mathcal{F}_2(\gamma)$ which select the same minimum outgoing edge, a contradiction.
\end{proof}

\begin{lemma}
\label{lem:compute_newp}
Let any fragment $F, F \not \in \mathcal{F}_3(\gamma),$ of a configuration $\gamma \in \Gamma_{\tt CF}$. If $F \in \mathcal{F}_2(\gamma)$ then every computation suffix starting from $\gamma$ contains a configuration $\gamma' \in \Gamma_{\tt CF}$ such that $F \in \mathcal{F}_3(\gamma')$.
\end{lemma}

\begin{proof}
Assume, by the contradiction, that there exists a suffix $e''$ starting from $\gamma$ with no configuration $\gamma' \in \Gamma_{\tt CF}$ such that $F \in \mathcal{F}_3(\gamma')$ in computation $e=e'e''$. Consider the configuration $\gamma$. Since $F \in \mathcal{F}_2(\gamma)$ and $F \not \in \mathcal{F}_3(\gamma)$, only Rule $\RF$ could be enabled at a node $v \in F$ (by definition of $\mathcal{F}_2(\gamma)$ and according to the guards of rules given in the formal description of Algorithm $\MSTA$). Moreover, as $F \not \in \mathcal{F}_3(\gamma)$ there exists at least one node $v \in F$ such that Predicate $ChangeNewP(v)$ is satisfied at $v$. Consider a computation step $\gamma \mapsto \gamma''$ of $e$. Assume that Rule $\RF$ is enabled at $v$ in $\gamma$ and not in $\gamma''$ but $v$ did not execute Rule $\RF$. If $v$ is the root of $F$ (i.e., $\parent_v=\emptyset$) or $v$ is on the path between the root of $F$ and the selected minimum-weight outgoing edge then $\neg ChangeNewP(v)$ implies that $newp_v=NewParent(v)$, a contradiction since Rule $\RF$ is the only rule in $\gamma$ which can change variable $newp_v$. Otherwise, for every other node $v \in F$, $\neg ChangeNewP(v)$ implies that $newp_{\parent_v}=\parent_v$, a contradiction since Rule $\RF$ is the only rule in $\gamma$ which can change variable $newp_v$. By weakly-fairness assumption on the daemon, every node $v \in F$ executes Rule $\RF$ and satisfies $\neg ChangeNewP(v)$.

Finally, we can observe that the set of fragments $\mathcal{F}_3(\gamma)$ is included in the set $\mathcal{F}_2(\gamma)$ by definition in a configuration $\gamma$.
\end{proof}

\begin{lemma}
\label{lem:time_mwoe}
Let any fragment $F \in \mathcal{F}_2(\gamma)$ in a configuration $\gamma \in \Gamma_{\tt CF}$. In $O(h_F)$ rounds, we have $F \in \mathcal{F}_3(\gamma')$, with $\gamma' \in \Gamma_{\tt CF}$ and $h_F$ the height of $F$.
\end{lemma}

\begin{proof}
Let $d_F(v)$ denotes the height of $v$ in $F$. We show by induction the following proposition: In at most $O(j+1)$ rounds, we have $\forall v \in F, d_F(v) \leq j \Rightarrow ([(\parent_v=\emptyset \vee \m_{\parent_v}=\m_v) \Rightarrow newp_v=NewParent(v)] \vee newp_v=\parent_v)$.

In base case $j=0$. Consider the root $v$ of Fragment $F$ (i.e., $\parent_v=\emptyset$). If $newp_v \neq NewParent(v)$ then Rule $\RF$ is enabled at $v$ in round 0, since $(\parent_v=\emptyset \wedge newp_v \neq NewParent(v)) \Rightarrow ChangeNewP(v)$. Therefore, since the daemon is weakly fair then in the first configuration of round 1, $v$ executes Rule $\RF$ and we have $newp_v=NewParent(v)$ at $v$ which verifies the proposition.

Induction case: We assume that in round $j=h_F-1$ we have $\forall u \in F, d_F(v) \leq j \Rightarrow ([(\parent_v=\emptyset \vee \m_{\parent_v}=\m_v) \Rightarrow newp_v=NewParent(v)] \vee newp_v=\parent_v)$. We have to show that in round $j+1$ we have $\forall v \in F, d_F(v) \leq j+1 \Rightarrow ([(\parent_v=\emptyset \vee \m_{\parent_v}=\m_v) \Rightarrow newp_v=NewParent(v)] \vee newp_v=\parent_v)$. Consider any node $v \in F$ of height $j+1$ in $F$. By induction hypothesis, we have either $newp_{\parent_v}=NewParent(\parent_v)=v$ or $\m_{\parent_v} \neq \m_v$. In the former case, if $newp_v \neq NewParent(v)$ and $\m_v=\m_{\parent_v}$ then Rule $\RF$ is enabled at $v$ in round $j$ (because $(newp_{\parent_v}=NewParent(\parent_v)=v \wedge \m_v=\m_{\parent_v} \wedge newp_v \neq NewParent(v)) \Rightarrow ChangeNewP(v)$). In the latter case, if $newp_v \neq \parent_v$ and $\m_v \neq \m_{\parent_v}$ then Rule $\RF$ is enabled at $v$ in round $j$ (because $(\m_v \neq \m_{\parent_v} \wedge newp_v \neq \parent_v) \Rightarrow ChangeNewP(v)$). Thus, since the daemon is weakly fair then in the first configuration of round $j+1$ $v$ executes Rule $\RF$. So, we have $([(\parent_v=\emptyset \vee \m_{\parent_v}=\m_v) \Rightarrow newp_v=NewParent(v)] \vee newp_v=\parent_v)$ at $v$. Therefore, in at most $O(h_F)$ rounds we have $\forall v \in F, d_F(v) \leq h_F \Rightarrow ([(\parent_v=\emptyset \vee \m_{\parent_v}=\m_v) \Rightarrow newp_v=NewParent(v)] \vee [\m_{\parent_v} \neq \m_v \wedge newp_v=\parent_v])$ and this implies that $\forall v \in F, \neg ChangeNewP(v)$.
\end{proof}

\begin{lemma}
\label{lem:merging_two_frag_dist}
Let any two fragments $F_1$ and $F_2$, $F_1,F_2 \not \in \mathcal{F}_4(\gamma),$ of a configuration $\gamma \in \Gamma_{\tt CF}$. If $F_1,F_2 \in \mathcal{F}_3(\gamma)$ and the same minimum-weight outgoing edge is selected by the two fragments then every computation suffix starting from $\gamma$ contains a configuration $\gamma' \in \Gamma_{\tt CF}$ such that $F_1,F_2 \in \mathcal{F}_4(\gamma')$.
\end{lemma}

\begin{proof}
Assume, by the contradiction, that there exists a suffix $e''$ starting from $\gamma$ with no configuration $\gamma' \in \Gamma_{\tt CF}$ such that $F_1,F_2 \in \mathcal{F}_4(\gamma')$ in computation $e=e'e''$.  As Rule $\RDist$ is the only rule to modify variable $newd_v$ such that $\neg ChangeNewD(v)$ is satisfied when executed, this implies that there exists a node $v \in (F_1 \cup F_2)$ which never executes Rule $\RDist$ in the computation suffix $e''$. Consider the configuration $\gamma$. According to the formal description of Algorithm \MSTA, Rules $\RCo, \RSC, \RLC, \RMin, \RF$, and $\RC$ are disabled for any node $v \in (F_1 \cup F_2)$ since $F_1,F_2 \in \mathcal{F}_3(\gamma)$. Moreover, Rule $\RFEnd$ is disabled at any node $v \in (F_1 \cup F_2)$ as we have $d_u \neq newd_u$ for every node $u \in NewChild(v)$, because $ F_1,F_2 \not \in \mathcal{F}_4(\gamma)$ and $newd_v=\infty$ by the execution of Rule $\RF$. So, only Rule $\RDist$ could be enabled at every node $v \in (F_1 \cup F_2)$. This implies that Rule $\RDist$ is disabled for every node $v \in (F_1 \cup F_2)$ (i.e., we have $\neg ChangeNewD(v)$). Consider without loss of generality Fragment $F=F_1$. Note that $F \in \mathcal{F}_3(\gamma), F \not \in \mathcal{F}_4(\gamma),$ and $\forall v \in F$ we have $newd_v=\infty$ due to the execution of Rule $\RF$. If $v \in F$ is the future root of $F$ (after the merging phase) then either $newp_{newp_v} \neq v$, a contradiction because $F_1$ and $F_2$ have selected the same minimum-weight outgoing edge, or $newd_v \leq 1$, a contradiction since $newd_v \neq \infty$. If $v \in F$ is any other node in $F$ then we have $(newd_v=newd_{newp_v}+1) \Rightarrow \neg ChangeNewD(v)$, a contradiction since $newd_v \neq \infty$. Therefore, the system reaches a configuration $\gamma' \in \Gamma_{\tt CF}$ in which for every node $v \in F$ we have $\neg ChangeNewD(v)$, so $F \in \mathcal{F}_4(\gamma')$.

Finally, we can observe that the set of fragments $\mathcal{F}_4(\gamma)$ is included in the set $\mathcal{F}_3(\gamma)$ by definition in a configuration $\gamma$.
\end{proof}

\begin{lemma}
\label{lem:time_F4}
Let any fragment $F \in \mathcal{F}_3(\gamma)$ in a configuration $\gamma \in \Gamma_{\tt CF}$. In $O(h_F)$ rounds, we have $F \in \mathcal{F}_4(\gamma')$, with $\gamma' \in \Gamma_{\tt CF}$ and $h_F$ the height of $F$.
\end{lemma}

\begin{proof}
We can show by induction on the height of $F$ that in $O(h_F)$ rounds every node $v \in F$ satisfies $\neg NewChangeD(v)$) using the same method as in proof of Lemma~\ref{lem:time_mwoe}.
\end{proof}
\begin{lemma}
\label{lem:merging_two_frag_end}
Let any fragment $F, F \not \in \mathcal{F}_5(\gamma),$ of a configuration $\gamma \in \Gamma_{\tt CF}$. If $F \in \mathcal{F}_4(\gamma)$ then every computation suffix starting from $\gamma$ contains a configuration $\gamma' \in \Gamma_{\tt CF}$ such that $F \in \mathcal{F}_5(\gamma')$.
\end{lemma}

\begin{proof}
Assume, by the contradiction, that there exists a suffix $e''$ starting from $\gamma$ with no configuration $\gamma' \in \Gamma_{\tt CF}$ such that $F \in \mathcal{F}_3(\gamma')$ in computation $e=e'e''$. Consider the configuration $\gamma$. Since $F \in \mathcal{F}_4(\gamma)$, only Rule $\RFEnd$ could be enabled at a node $v \in F$ (by definition of $\mathcal{F}_4(\gamma)$ and according to the guards of rules given in the formal description of Algorithm $\MSTA$). Moreover, as $F \not \in \mathcal{F}_5(\gamma)$ there exists at least one node $v \in F$ such that Predicate $CopyVar(v)$ is satisfied at $v$. Consider a computation step $\gamma \mapsto \gamma''$ of $e$. Assume that Rule $\RFEnd$ is enabled at $v$ in $\gamma$ and not in $\gamma''$ but $v$ did not execute Rule $\RFEnd$. If $v \in F$ is a leaf and $v$ is not adjacent to the minimum-weight outgoing edge of $F$ then $\neg CopyVar(v)$ implies that $\parent_v=newp_v \wedge d_v=newd_v$ (because $v$ has no neighbor $u \in Child(v)$), a contradiction since Rule $\RFEnd$ is the only rule which could copy the value of $newp_v$ (resp. $newd_v$) in $\parent_v$ (resp. $d_v$). Otherwise for every other node $v \in F$, $\neg CopyVar(v)$ implies that either $\parent_v=newp_v \wedge d_v=newd_v$ or $\exists u \in NewChild(v)$ such that  $d_u \neq newd_u \vee \parent_u \neq newp_u)$. In the former case, there is a contradiction since Rule $\RFEnd$ is the only rule which could copy the value of $newp_v$ (resp. $newd_v$) in $\parent_v$ (resp. $d_v$). In the latter case, there is a neighbor $u \in NewChild(v)$ which modified its variable $\parent_u$ or $d_u$ by executing Rule $\RCo$ or $\RC$, a contradiction since only Rule $\RFEnd$ could be enabled at a node $v \in F$ in $\gamma$. By weakly-fairness assumption on the daemon, every node $v \in F$ executes Rule $\RFEnd$ and satisfies $\neg CopyVar(v)$.

Finally, we can observe that the set of fragments $\mathcal{F}_5(\gamma)$ is included in the set $\mathcal{F}_4(\gamma)$ by definition in a configuration $\gamma$.
\end{proof}

\begin{lemma}
\label{lem:time_F5}
Let any fragment $F \in \mathcal{F}_4(\gamma)$ in a configuration $\gamma \in \Gamma_{\tt CF}$. In $O(h_F)$ rounds, we have $F \in \mathcal{F}_5(\gamma')$, with $\gamma' \in \Gamma_{\tt CF}$ and $h_F$ the height of $F$.
\end{lemma}

\begin{proof}
We can show by induction on the height of $F$ that in $O(h_F)$ rounds every node $v \in F$ satisfies $\neg CopyVar(v)$) using the same method as in proof of Lemma~\ref{lem:time_mwoe}.
\end{proof}

\begin{lemma}
\label{lem:nb_frag_diminution}
Let a configuration $\gamma \in \Gamma_{\tt CF}$ such that $|\mathcal{F}(\gamma)|>1$. We have $|\mathcal{F}(\gamma')|<|\mathcal{F}(\gamma)|$ for any configuration $\gamma' \in \Gamma_{\tt CF}$ obtained after a merging step in every computation suffix starting from $\gamma$.
\end{lemma}

\begin{proof}
Assume, by the contradiction, that there exists a suffix $e''$ starting from $\gamma$ with a configuration $\gamma' \in \Gamma_{\tt CF}$ obtained after a merging step for which $|\mathcal{F}(\gamma')| \geq |\mathcal{F}(\gamma)|$ in computation $e=e'e''$. This implies that in $e''$ either there are no two fragments $F_1,F_2 \in \mathcal{F}(\gamma)$ which can merge together using the same minimum-weight outgoing edge, or $F_1$ and $F_2$ does not belong to the same fragment after a merging step. First of all, by Lemmas~\ref{lem:labelMST} and~\ref{lem:Min} every fragment $F \in \mathcal{F}(\gamma)$ which does not satisfies $CorrectF(v)$ and $\m_v[0]=\Fusion(v)$ executes Rules $\RSC, \RLC$ and $\RMin$ to belong to $\mathcal{F}_2(\gamma)$. So, we consider that every fragment in $\mathcal{F}(\gamma)$ belongs to $\mathcal{F}_2(\gamma)$. In the first case, this is a contradiction with Lemma~\ref{lem:merging_two_frag} which shows that in $\gamma$ there are at least two fragments $F_1$ and $F_2$, $F_1,F_2 \in \mathcal{F}_2(\gamma)$ which select the same minimum-weight outgoing edge for merging. In the latter case, by Lemma~\ref{lem:compute_newp} every fragment $F \in \mathcal{F}_2(\gamma)$ computes its future parent after the merging step, so $F_1,F_2 \in \mathcal{F}_3(\gamma)$. Moreover, by Lemma~\ref{lem:merging_two_frag_dist} we have $F_1,F_2 \in \mathcal{F}_4(\gamma')$ in $e''$. So, by Lemma~\ref{lem:merging_two_frag_end} every node $v \in (F_1 \cup F_2)$ can execute Rule $\RFEnd$ and then Rules $\RSC$ and $\RLC$ to form a new fragment in $\mathcal{F}(\gamma')$ composed of $F_1$ and $F_2$ in $e''$, a contradiction.
\end{proof}

Note that for each fragment $F \in \mathcal{F}_5(\gamma)$ in $\gamma \in \Gamma_{\tt CF}$, we have that each node $v \in F$ satisfies Predicate $\Distance(v)$ but does not satisfies Predicate $CorrectF(v)$ (because of the merging step the labels of each node $v$ is no more correct). So, every node $v \in F$ can execute again Rules $\RSC$ and $\RLC$.

\begin{lemma}
Let any fragment $F \in \mathcal{F}_5(\gamma)$ of a configuration $\gamma \in \Gamma_{\tt CF}$ such that $\forall v \in F, CorrectF(v)$. If $F$ does not span all the nodes of the system, then Rule $\RMin$ is enabled in at least one node $v \in F$ in $\gamma$.
\end{lemma}

\begin{proof}
First of all, Predicate $CorrectF(v)$ is satisfied at every node $v \in F$ because $F \in \mathcal{F}_5$. Assume, by the contradiction, that $F$ does not spans all the nodes of the system in $\gamma$ and Rule $\RMin$ is disabled at every node $v \in F$. This implies that there is a node $v \in F$ which adjacent to an edge $(u,v)$ such that $u \not \in F$. So, we have $\FA(v) \neq \emptyset \Rightarrow \Fusion(v) \neq \emptyset$. Moreover, since $F \in \mathcal{F}_5$, $v$ has executed Rule $\RFEnd$ and we have $\m_v=\emptyset$. Therefore, Rule $\RMin$ is enabled at $v$, a contradiction.
\end{proof}

\begin{lemma}
\label{lem:time_merging}
Starting from any configuration $\gamma \in \Gamma$ which contains several fragments, the system reaches a configuration $\gamma' \in \Gamma_{\tt CF}$ which contains a single fragment spanning all the nodes of the system in $O(n^2)$ rounds, with $n$ the network size.
\end{lemma}

\begin{proof}
First of all, according to Lemmas~\ref{lem:time_cf} and~\ref{lem:labelMST} in $O(n)$ rounds the system reaches a configuration $\gamma_1 \in \Gamma_{\tt CF}$ in which each fragment $F \in \mathcal{F}_1(\gamma_1)$. Moreover, according to Lemmas~\ref{lem:min_convergence},~\ref{lem:time_mwoe},~\ref{lem:time_F4} and~\ref{lem:time_F5} by summing up the complexities each merging step is performed using at most $O(n)$ rounds (since $n$ is an upper bound for the height of any fragment). Finally, in a configuration we can not have more than $n$ fragments so to obtain a spanning tree we can perform at most $n-1$ merging phases (attained when only two fragments can be merged at each step). Moreover, by Lemma~\ref{lem:nb_frag_diminution} after each merging step the number of fragments is decreased by at least one. Therefore, by the above elements we obtain that a spanning tree is constructed in $O(n^2)$ rounds starting from an arbitrary configuration.
\end{proof}

\begin{lemma}
\label{lem:disable_merge_rules}
In every configuration $\gamma \in \Gamma_{\tt CF}$ which contains a single fragment $T \in \mathcal{F}_5(\gamma)$ spanning all the nodes of the system and correctly labeled, then for every node $v \in T$ no rule of Algorithm $\MSTA$, except Rule $\RC$, is enabled at $v$.
\end{lemma}

\begin{proof}
Assume, by the contradiction, that there is an enabled rule, except Rule $\RC$, of Algorithm $\MSTA$ in a node $v \in T$ in $\gamma$. For every node $v \in T$ Predicate $\Distance(v)$ is satisfied in $\gamma \in \Gamma_{\tt CF}$ (by definition of $\Gamma_{\tt CF}$), so Rule $\RCo$ is disabled at $v$, a contradiction. Since $T$ is correctly labeled, we have $\neg CorrectF(v) \Rightarrow (\SizeC(v) \wedge \Label(v))$, so Rules $\RSC$ and $\RLC$ are disabled at $v$, a contradiction. Since $T$ spans all the nodes of the system, for every node $v \in T$ we have $\Fusion(v)=\emptyset$ and Rule $\RMin$ is disabled at $v$, a contradiction. Finally, we have that $T \in \mathcal{F}_5(\gamma)$ and by definition of $\mathcal{F}_5(\gamma)$ Rules $\RF,\RDist,\RFEnd$, and $\RC$ are disabled at $v$, a contradiction.
\end{proof}

\subsubsection{Correctness and complexity of the recovering phase}

\begin{lemma}
Let a configuration $\gamma \in \Gamma_{\tt CF}$ such that $|\mathcal{F}(\gamma)|=1$ and $F$ be the fragment of $\mathcal{F}(\gamma)$. At least one node $v \in F$ can execute Rule $\RC$.
\end{lemma}

\begin{proof}
First of all, we consider that every node $v \in F$ the distance and the label are correct (i.e., $CorrectF(v)$ and $\neg CopyVar(v)$ are satisfied which implies we have $\parent_v=newp_v \wedge d_v=newd_v \wedge newd_v=new_{\parent_v}+1$). Assume, by the contradiction, that Rule $\RC$ is disabled for every node $v \in F$. This implies that for every node $v \in F$ we have either $\Fusion(v) \neq \emptyset$, or another internal edge can not be selected. In the first case, by hypothesis of the lemma there is only one fragment in $\mathcal{F}(\gamma)$, so for every node $v \in F$ we have $\Fusion(v) = \emptyset$, a contradiction. In the second case, this implies that an internal edge of minimum-weight associated to the common ancestor of the edge is not propagated up in $F$. If $v$ is the common ancestor of a locally internal edge (given by Macro $\IE_l(v)$), then we have $EndForward(v) \Rightarrow Recover(v)$ and Rule $\RC$ is enabled at $v$, a contradiction. If $v$ has selected a local internal edge (i.e., $\mIns_v \in \IE_l(v)$ and $\mIns_v \not \in \IE_c(v)$) and the internal edge has been propagated by $\parent_v$ ($\mIns_v=\mIns_{\parent_v}$), then this implies we have $Forwarded(v) \Rightarrow Recover(v)$ and Rule $\RC$ is enabled at $v$, a contradiction. Otherwise, for an internal edge propagated by a child $u$ in $F$ either $v$ is the common ancestor, or the internal edge has been propagated by $\parent_v$ ($\mIns_v=\mIns_{\parent_v}$) and $u$ has selected another internal edge (i.e., $\mIns_v \not \in \IE_c(v)$). This implies we have $Forwarded(v) \Rightarrow Recover(v)$ and Rule $\RC$ is enabled at $v$, a contradiction.
\end{proof}

\begin{corollary}
\label{cor:propag_internal_edge}
In any configuration $\gamma \in \Gamma_{\tt CF}$ such that $|\mathcal{F}(\gamma)|=1$, by executing Rule $\RC$ every node $v \in F$ sends up in $F$ the internal edges selected by $v$ and its descendants ordered locally on the nearest common ancestors (given by Macro $\IE(v)$), with $F$ the fragment in $\mathcal{F}(\gamma)$.
\end{corollary}

\begin{lemma}
\label{lem:end_propag_ie}
Let a configuration $\gamma \in \Gamma_{\tt CF}$ such that $|\mathcal{F}(\gamma)|=1$ and $F$ be the fragment of $\mathcal{F}(\gamma)$. Every internal edge related to the nearest common ancestor $x$ is not propagated up in $F$.
\end{lemma}

\begin{proof}
According to Corollary~\ref{cor:propag_internal_edge}, $x$ propagates up in $F$ the internal edges selected by its descendants and itself. Assume, by the contradiction, that the parent $x$ of the nearest common ancestor $y$ related to an internal edge $e \in E$ propagates up in $F$ the internal edge $e=(u,v)$, described by its weight $w(u,v)$ and the labels of its extremities $\lab_u$ and $\lab_v$ stored in variable $\mIns_y$ at $y$. This implies that when $x$ executes Rule $\RC$ then Macro $\IE(x)$ returns edge $e$. Macro $\IE(x)$ returns an edge from the union set of edges given by Macros $\IE_l(x)$ and $\IE_c(x)$. We must consider only Macro $\IE_c(x)$ since $y$ is the common ancestor of $e$. However, according to the formal description of Algorithm $\MSTA$ Macro $\IE_c(x)$ contains only internal edges $f=(a,b)$ whose the nearest common ancestor related to $f$ has a label higher or equal to $x$'s label following the lexicographical order (i.e., $\Lca(\lab_a,\lab_b) \succeq \lab_x$). Thus, we have $e \not \in \IE(x)$, a contradiction.
\end{proof}

\begin{lemma}
\label{lem:create_new_frag}
Let a configuration $\gamma \in \Gamma_{\tt CF}$ such that $|\mathcal{F}(\gamma)|=1$ and $F$ be the fragment of $\mathcal{F}(\gamma)$. If a node $v \in F$ selects by executing Rule $\RC$ an internal edge $e=(x,y)$ such that $w(x,y)<w(v,\parent_v)$ and $v$ is not the common ancestor related to $(x,y)$ then the edge $(v,\parent_v)$ is deleted by $v$ from $F$.
\end{lemma}

\begin{proof}
Assume, by the contradiction, that $v \in F$ selects by executing Rule $\RC$ an internal edge $e=(x,y)$ such that $w(x,y)<w(v,\parent_v)$ but the edge $(v,\parent_v)$ is not deleted from $F$. According to the formal description of Algorithm $\MSTA$, a node can delete an edge of a fragment by executing Rule $\RC$. So, this implies that by executing Rule $\RC$ the edge $(v,\parent_v)$ is not deleted from $F$ by $v \in F$. Consider the internal edge $e=(x,y)$ stored in $\mIns_v$ by executing Rule $\RC$ at $v$ such that $\mIns_v[0]<w(v,\parent_v)$. By description of Rule $\RC$, $v$ does not delete the edge $(v,\parent_v)$ only if $\Lca(\mIns_v[1],\mIns_v[2])=\lab_v)$, which is a contradiction with the hypothesis of the lemma.
\end{proof}

\begin{lemma}
\label{lem:time_recover}
Starting from any configuration $\gamma \in \Gamma_{\tt CF}$ which contains a single spanning tree $T$, the recovering phase is performed in $O(n^2)$ rounds, $n$ the network size.
\end{lemma}

\begin{proof}
First of all, every node $v \in T$ sends up in $T$ the internal edge of minimum weight associated to each common ancestor $ca$, given by Macro $\IE(v)$ based on Macro $minIE(v,ca)$. Moreover, by Lemma~\ref{lem:end_propag_ie} every internal edge $e$ is not propagated by the ancestors of the common ancestor $e$ in $T$. By Lemma~\ref{cor:propag_internal_edge}, every node $v \in T$ sends up in the tree the internal edges selected by $v$ and its descendants ordered locally on the nearest common ancestors, that is following the lexicographical order on the label of nearest common ancestors. Observe that every node $v \in T$ is the common ancestor of at most $h_T$ internal edges selected to be propagated up in $T$, with $h_T$ the height of $T$. Furthermore, each propagated internal edge reaches its related nearest common ancestor in $O(h_T)$ rounds. However, the propagation of the internal edges is pipelined in $T$, since a node $v \in T$ can execute Rule $\RC$ when its parent propagates its internal edge or the nearest common ancestor is reached (see Predicate $SelectEdge(v)$). Thus, for every nearest common ancestor $v \in T$ the propagation of the internal edges related to $v$ is performed in $O(h_T)$ rounds. Finally, there are at most $n$ nearest common ancestors in the spanning tree $T$, so the propagation of all the internal edges of $T$ is performed in $O(n.h_T) \leq O(n^2)$ rounds.
\end{proof}

\begin{lemma}
\label{lem:mst_preserved}
Starting from every configuration $\gamma \in \Gamma_{\tt CF}$ satisfying Definition~\ref{def:mst}, the system can only reach a configuration $\gamma' \in \Gamma_{\tt CF}$ which satisfies Definition~\ref{def:mst}.
\end{lemma}

\begin{proof}
By Lemma~\ref{lem:disable_merge_rules}, in every configuration $\gamma \in \Gamma_{\tt CF}$ every rule of Algorithm $\MSTA$, except Rule $\RC$, is disabled at $v \in V$. Consider any configuration $\gamma \in \Gamma_{\tt CF}$ which satisfies Definition~\ref{def:mst}. This implies that there is only a single spanning tree $T$ in $\gamma$ and in every fundamental cycle defined by each internal edge of $T$ Lemma~\ref{lem:create_new_frag} can not be applied. Therefore, by executing Rule $\RC$ at any node $v \in T$ no new fragment is created and the constructed minimum spanning tree $T$ is preserved.
\end{proof}

\begin{theorem}
Algorithm $\MSTA$ is a self-stabilizing algorithm for Specification~\ref{spec:mst} under a weakly fair daemon with a convergence time of $O(n^2)$ rounds and memory complexity of $O(\log^2 n)$ bits per node, with $n$ the network size.
\end{theorem}

\begin{proof}
We have to show first that starting from any configuration the execution of Algorithm $\MSTA$ verifies Property [TC1] and [TC2] of Specification~\ref{spec:mst}.

First of all, by Theorem~\ref{thm:no_deadlock} while the system does not reach a configuration satisfying Definition~\ref{def:mst}, there is a rule enabled, except Rule $\RC$, at a node $v \in V$. According to Lemmas~\ref{lem:nb_frag_diminution},~\ref{lem:create_new_frag},~\ref{lem:time_merging} and~\ref{lem:time_recover}, from any configuration Algorithm $\MSTA$ reaches a configuration $\gamma \in \Gamma$ satisfying Definition~\ref{def:mst} in finite time, which verifies Property [TC1]. Moreover, according to Lemma~\ref{lem:mst_preserved} from a configuration $\gamma \in \Gamma$ satisfying Definition~\ref{def:mst} Algorithm $\MSTA$ can only reach a configuration in $\Gamma$ satisfying Definition~\ref{def:mst}, which verifies Property [TC2] of Specification~\ref{spec:mst}.

We consider now the convergence time and memory complexity of Algorithm $\MSTA$. According to Lemmas~\ref{lem:time_merging} and~\ref{lem:time_recover}, each part of the algorithm (merging and recovering part) have a convergence time of at most $O(n^2)$ rounds to construct a minimum spanning tree. Moreover, Algorithm $\MSTA$ maintains height variables at every node $v \in V$, composed of six variables of size $\log(n)$ bits (variables $\parent_v, d_v, newp_v, newd_v, \size_v$, and $\m_v$) and two variables of size $O(\log^2(n))$ bits used to stored labels of nodes (variables $\lab_v$ and $\mIns_v$). According to \cite{Peleg00}, Variable $\lab_v$ necessitates $\Theta(\log^2 n)$ bits of memory at every node $v \in V$. Therefore, no more than $O(\log^2(n))$ bits per node are necessary.
\end{proof}

\section{Conclusion}
We extended the Gallager, Humblet and Spira (GHS) algorithm, \cite{GallagerHS83}, to self-stabilizing settings via a compact informative labeling scheme. Thus, the resulting solution presents several advantages appealing for large scale systems: it is compact since it uses only poly-logarithmic in the size of the network memory space ($O(\log^2(n))$ bits per node) and it scales well since it does not rely on any global parameter of the network. The convergence time of the proposed solution is $O(n^2)$ rounds. Quite recently, another self-stabilizing algorithm was proposed by Korman et al.~\cite{KormanKM11} for the MST problem with a convergence time of $O(n)$ rounds and memory complexity of $O(\log(n))$ bits. However, this approach requires the use of several sub-algorithms leading to a complex solution to be used in a practical situation, comparing to our algorithm.



\begin{thebibliography}{10}
\bibitem{AGKR02}
Stephen Alstrup, Cyril Gavoille, Haim Kaplan  and  Theis Rauhe .
\newblock Nearest common ancestors: a survey and a new algorithm for a distributed environment.
\newblock {\em Theory of Computing Systems}, 37(3):441--456, 2004.

\bibitem{BeinDV05}
Doina Bein, Ajoy Kumar Datta and Vincent Villain.
\newblock Self-Stablizing Pivot Interval Routing in General Networks.
\newblock {\em ISPAN}, pages 282--287, 2005.

\bibitem{BlinDGR10}
L{\'e}lia Blin, Shlomi Dolev, Maria Gradinariu Potop-Butucaru and Stephane Rovedakis
\newblock Fast Self-stabilizing Minimum Spanning Tree Construction - Using Compact Nearest Common Ancestor Labeling Scheme.
\newblock {\em 24th International Symposium on Distributed Computing (DISC)}, volume 6343 of {\em Lecture Notes in Computer Science}, pages 480--494, 2010.

\bibitem{BPRT09c}
L{\'e}lia Blin, Maria Potop-Butucaru, Stephane Rovedakis and S{\'e}bastien Tixeuil.
\newblock  A New Self-stabilizing Minimum Spanning Tree Construction with Loop-Free Property.
\newblock {\em 23rd International Symposium on Distributed Computing (DISC)}, volume 5805 of {\em Lecture Notes in Computer Science}, pages 407--422. Springer 2009.

\bibitem{D74j}
Edsger~W. Dijkstra.
\newblock Self-stabilizing systems in spite of distributed control.
\newblock {\em Commun. ACM}, 17(11):643--644, 1974.

\bibitem{Dolev00}
Shlomi Dolev.
\newblock {\em Self-Stabilization}.
\newblock MIT Press, 2000.

\bibitem{BK07}
Janna Burman and Shay Kutten.
\newblock Time Optimal Asynchronous Self-stabilizing Spanning Tree.
\newblock {\em 21st International Symposium on Distributed Computing (DISC)}, volume 4731 of {\em Lecture Notes in Computer Science}, pages 92-107. Springer 2007.

\bibitem{DolevIM97}
Shlomi Dolev, Amos Israeli and Shlomo Moran,
\newblock Uniform Dynamic Self-Stabilizing Leader Election.
\newblock IEEE Trans. Parallel Distrib. Syst., volume 8-4,  pages 424--440, 1997.

\bibitem{GallagerHS83}
Robert~G. Gallager, Pierre~A. Humblet, and Philip~M. Spira.
\newblock A distributed algorithm for minimum-weight spanning trees.
\newblock {\em ACM Trans. Program. Lang. Syst.}, 5(1):66--77, 1983.

\bibitem{AntonoiuS97}
Gheorghe Antonoiu and Pradip K. Srimani.
\newblock Distributed Self-Stabilizing Algorithm for Minimum Spanning Tree Construction.
\newblock {\em 3rd International Conference on Parallel and Distributed Computing (Euro-Par)}, volume 1300 of {\em Lecture Notes in Computer Science}, pages 480-487. Springer 1997.

\bibitem{GuptaS03}
Sandeep K.~S. Gupta and Pradip~K. Srimani.
\newblock Self-stabilizing multicast protocols for ad hoc networks.
\newblock {\em J. Parallel Distrib. Comput.}, 63(1):87--96, 2003.

\bibitem{HT84j}
D. Harel and R. E. Tarjan.
\newblock Fast algorithms for finding nearest common ancestors. 
\newblock {\em SIAM Journal Computing}, 13(2):338-355, 1984. 

\bibitem{HighamL01}
Lisa Higham and Zhiying Liang.
\newblock Self-stabilizing minimum spanning tree construction on message-passing networks.
\newblock In {\em 15th International Conference on Distributed Computing (DISC)}, volume 2180 of {\em  Lecture Notes in Computer Science}, pages 194--208, 2001.

\bibitem{KP93a}
S~Katz and KJ~Perry.
\newblock Self-stabilizing extensions for message-passing systems.
\newblock {\em Distributed Computing}, 7:17--26, 1993.

\bibitem{KormanK07}
Amos Korman and Shay Kutten.
\newblock Distributed verification of minimum spanning trees.
\newblock In {\em Distributed Computing},  20(4):   pages   253--266, 2007.

\bibitem{KormanKM11}
Amos Korman and Shay Kutten and Toshimitsu Masuzawa.
\newblock Fast and compact self stabilizing verification, computation, and fault detection of an MST.
\newblock In {\em 30th Annual ACM Symposium on Principles of Distributed Computing (PODC)}, pages 311--320, 2011. 

\bibitem{Kruskal56}
Joseph~B. Kruskal.
\newblock On the shortest spanning subtree of a graph and the travelling
  salesman problem.
\newblock {\em Proc. Amer. Math. Soc.}, 7:48--50, 1956.

\bibitem{ParkMHT90}
Jungho Park, Toshimitsu Masuzawa, Kenichi Hagihara and  Nobuki Tokura.
\newblock Distributed Algorithms for Reconstructing MST after Topology Change.
\newblock {\em 4th International Workshop on Distributed Algorithms (WDAG)}, pages 122--132, 1990.

\bibitem{ParkMHT92}
Jungho Park, Toshimitsu Masuzawa, Ken'ichi Hagihara and  Nobuki Tokura.
\newblock Efficient distributed algorithm to solve updating minimum spanning tree problem.
\newblock {\em Systems and Computers in Japan}, 23(3):1--12, 1992.

\bibitem{Peleg00}
David Peleg,
\newblock Informative Labeling Schemes for Graphs,
\newblock {\em MFCS}, pages 579--588, 2000.

\bibitem{Peleg00_book}
David Peleg,
\newblock {\em Distributed Computing: A Locality-Sensitive Approach}.
\newblock Society for Industrial and Applied Mathematics, 2000.

\bibitem{Prim57}
R.C. Prim.
\newblock Shortest connection networks and some generalizations.
\newblock {\em Bell System Tech. J.}, pages 1389--1401, 1957.

\bibitem{tar83}
R. E. Tarjan,
\newblock Data Structures and Network Algorithms.
\newblock SIAM, volume 44, 1983.

\bibitem{Tel94}
Gerard Tel.
\newblock {\em Introduction to distributed algorithm}.
\newblock Cambridge University Press, Second edition, 2000.

\bibitem{GoudaM91}
Mohamed G. Gouda and Nicholas J. Multari.
\newblock Stabilizing Communication Protocols.
\newblock IEEE Trans. Computers, volume 40(4), pages 448-458, 1991.

\bibitem{FlocchiniEPPS07}
Paola Flocchini and Toni Mesa Enriquez and Linda Pagli and Giuseppe Prencipe and Nicola Santoro.
\newblock Distributed Computation of All Node Replacements of a Minimum Spanning Tree.
\newblock {\em 13th International Conference on Parallel and Distributed Computing (Euro-Par)}, volume 4641 of {\em Lecture Notes in Computer Science}, pages 598-607. Springer 2007.

\bibitem{FlocchiniEPPS12}
Paola Flocchini and Toni Mesa Enriquez and Linda Pagli and Giuseppe Prencipe and Nicola Santoro.
\newblock Distributed Minimum Spanning Tree Maintenance for Transient Node Failures.
\newblock {\em IEEE Trans. Computers}, 61(3):408--414, 2012.

\end{thebibliography}
\end{document}